\pdfoutput=1
\documentclass[10pt]{article}
\usepackage[preprint]{tmlr}

\usepackage{graphicx}
\usepackage{amsmath,amssymb,amsfonts}
\usepackage{amsthm}
\usepackage{xcolor}
\usepackage{booktabs}
\usepackage{array}
\usepackage{algorithm}
\usepackage{algorithmicx}
\usepackage{algpseudocode}
\usepackage{hyperref}
\usepackage{url}
\hypersetup{colorlinks=true,allcolors=blue!55!black}

\theoremstyle{plain}
\newtheorem{theorem}{Theorem}
\newtheorem{proposition}[theorem]{Proposition}
\newtheorem{corollary}[theorem]{Corollary}
\newtheorem{lemma}[theorem]{Lemma}
\theoremstyle{definition}
\newtheorem{assumption}{Assumption}
\newtheorem{remark}{Remark}

\title{Asynchronous, Generation-Free ABC: Streaming AMIS Reweighting for Heterogeneous Simulator Workloads on HPC}

\author{\name Julian Herold \email julian.herold@kit.edu \\
      \addr Scientific Computing Center, Karlsruhe Institute of Technology, Karlsruhe, Germany
      \AND
      \name Markus Götz \email markus.goetz@kit.edu \\
      \addr Scientific Computing Center, Karlsruhe Institute of Technology, Karlsruhe, Germany
      \AND
      \name Alexander Schug \email alexander.schug@kit.edu \\
      \addr Scientific Computing Center, Karlsruhe Institute of Technology, Karlsruhe, Germany}

\def\month{MM}  %
\def\year{YYYY} %
\def\openreview{\url{https://openreview.net/forum?id=XXXX}} %

\begin{document}

\maketitle

\begin{abstract}
Approximate Bayesian computation (ABC) enables Bayesian inference for models with intractable likelihoods by estimating parameter posterior distributions from simulated data. Sequential ABC methods improve sampling efficiency by iteratively adapting the proposal distribution toward regions of high posterior probability. However, these methods typically operate in generations, where each generation must wait for its slowest simulation before the next can begin. On a parallel machine, this synchronization barrier leaves workers idle and reduces throughput, particularly when simulation runtimes vary substantially.
We introduce a generation-free ABC sampler that updates its proposal from the best particles evaluated so far as simulation results arrive, immediately assigning a new candidate to each available worker. For posterior estimation, we reuse the full evaluation history and reweight particles against a reconstructed mixture of sampling proposals. 
We show how errors in approximating the mixture of sampling proposals affect the estimated posterior and establish consistency for the smooth-ABC posterior when the approximation becomes asymptotically exact. In our experiments, the method achieves well-calibrated posterior estimates where a pyABC baseline with the same kernel and bandwidth schedule under-covers.
From the asynchronous run's timing alone, we predict how much a synchronization barrier would slow the sampler. These predictions agree with measured synchronization penalties ranging from $1.2\times$ to $400\times$. On a Cellular Potts tissue simulator, the asynchronous sampler completes $2.1$--$2.4\times$ as many evaluations as pyABC in the same hour, and the discrepancy threshold for its best $100$ evaluations is $4.1\times$ smaller. With the starting bandwidth set from prior-predictive discrepancies, it also produces a more concentrated two-parameter posterior with a larger effective sample size than pyABC. 
\end{abstract}

\section{Introduction}\label{sec:introduction}

Approximate Bayesian computation (ABC) enables Bayesian inference when the likelihood is intractable but data can be generated using a simulator \citep{cranmer2020frontier,beaumont2019abc}. Sequential ABC methods, such as ABC sequential Monte Carlo (ABC-SMC) and ABC population Monte Carlo (ABC-PMC) \citep{sisson2007,toni2009abcsmc,beaumont2009adaptive,delmoral2012adaptive}, iteratively update a population of weighted parameter samples, to favor parameters whose simulated data better match the observations. Although simulations of individual candidates can run in parallel, the generation-based implementations considered here update their proposal and ABC bandwidth only after the entire population has been evaluated. Workers that finish early must therefore wait for the slowest simulation, which can substantially reduce throughput when runtimes vary. This is particularly relevant to cell-based tissue models, where different parameter values can produce very different numbers of cells and, consequently, simulation runtimes.

Our sampler avoids this barrier by updating the proposal whenever a simulation result becomes available. It constructs a mixture from the $k$ lowest-discrepancy particles evaluated so far and sends a new candidate to the free worker. Because proposals change as results arrive, the completed candidates come from different sampling distributions. We account for this using streaming adaptive multiple importance sampling (AMIS) \citep{cornuet2012amis}, reweighting the history against a mixture of the proposals used. The recorded evaluations also let us recompute posterior estimates after a run and examine the effect of approximating this mixture, without repeating expensive simulations.

We compare the sampler with pyABC \citep{klinger2018pyabc}, a generation-based ABC implementation configured with the same kernel and bandwidth schedule. The cost of synchronization is especially clear when one worker remains slow: in this experiment, pyABC throughput falls by more than a factor of one hundred, while our sampler loses less than half its throughput. We also investigate whether the additional evaluations improve posterior quality and at what simulation cost asynchronous updates become beneficial.

We make four contributions, evaluated in the corresponding parts of \S\ref{sec:results}:
\begin{enumerate}
\item A generation-free sampler that combines smooth-kernel ABC with streaming AMIS reweighting. Its posterior estimate can be reconstructed from the evaluation log (\S\ref{sec:method}).
\item A runtime-based prediction of the synchronization penalty. We validate it against a synchronized variant of the same sampler for measured synchronization penalties ranging from $1.2\times$ to $400\times$ (C1, \S\ref{sec:results-barrier}).
\item An empirical assessment of how increased throughput affects tolerance and when the benefits of asynchronous execution exceed its coordination overhead (C2, \S\ref{sec:results-converts}; C3, \S\ref{sec:results-boundary}).
\item A theoretical analysis of how approximating the mixture of sampling proposals affects posterior estimation, including conditions for consistency with the intended smooth-ABC posterior (\S\ref{sec:theory}). We also assess the quality of that approximation, posterior calibration and recovery, and the effects of starting bandwidth and archive size (C4, \S\ref{sec:results-posterior}).
\end{enumerate}

\section{Background and Related Work}\label{sec:background}\label{sec:background-abc}\label{sec:background-sequential}\label{sec:background-amis}\label{sec:background-distributed}
ABC compares simulated data with observations through a discrepancy rather than evaluating the likelihood directly. Let $y$ denote the observations, $\theta$ the model parameters, and $\rho$ the discrepancy between simulated and observed summaries. At bandwidth $\epsilon$, the resulting ABC target is
\begin{equation}
\pi_{\epsilon}(\theta \mid y) \propto \pi(\theta)\, L_{\epsilon}(\theta),
\qquad
L_{\epsilon}(\theta) = \mathbb{E}_{\rho \sim p(\rho \mid \theta)}\!\left[K_{\epsilon}(\rho)\right],
\label{eq:abc-target}
\end{equation}
with $\pi$ the prior and $K_\epsilon$ the ABC kernel. The classical kernel is the hard indicator $\mathbf{1}[\rho<\epsilon]$, under which a simulation is either accepted or rejected and the output is a population of accepted particles. With a smooth kernel, such as the Gaussian $\exp(-\rho^2/2\epsilon^2)$ \citep{wilkinson2013,fearnhead2012}, every simulation contributes a continuous weight rather than a binary decision. This also avoids the discontinuity introduced by hard acceptance thresholds. Throughout, $\epsilon$ is the bandwidth of this kernel and the \emph{tolerance} a sampler has reached is a separate, empirical quantity defined in \S\ref{sec:experiments-metrics}.

\paragraph{Sequential ABC and the generation.}
Sequential ABC improves sampling efficiency by adapting its proposal across generations. Methods in this family include ABC-SMC \citep{sisson2007,toni2009abcsmc}, ABC-PMC \citep{beaumont2009adaptive}, and subsequent adaptive variants \citep{delmoral2012adaptive,drovandi2011,lenormand2013apmc}. More recently, \citet{pesonen2025stratified} improved sampling efficiency by stratifying the discrepancy space. These sequential methods maintain populations of $N$ particles while progressively reducing the ABC bandwidth. Generation $t$ fits a mixture proposal $q_t$ to the weighted population of generation $t-1$, draws candidates from it until $N$ have been accepted at the current bandwidth, and weights each accepted particle by $\pi(\theta_i)/q_t(\theta_i)$. The next proposal needs the whole weighted population, so generation $t+1$ cannot begin before the last candidate of generation $t$ has been simulated. This creates the synchronization barrier described in \S\ref{sec:introduction}. Generation-based implementations such as pyABC \citep{klinger2018pyabc,schalte2022pyabc}, ABCpy \citep{dutta2017abcpy}, and ELFI \citep{lintusaari2018elfi} can parallelize the simulations within a generation but wait before constructing the next one.

\paragraph{Reusing every draw: AMIS.} When the proposal changes during sampling, particles collected at different stages come from different distributions. Adaptive multiple importance sampling (AMIS) \citep{cornuet2012amis} combines them by weighting each particle against a mixture of the proposals used. After $T$ stages with proposals $q_1,\dots,q_T$, the weight of every past particle is updated using this mixture,
\begin{equation}
w_i = \frac{\pi(\theta_i)}{\frac1T\sum_{t=1}^{T} q_t(\theta_i)}
\qquad\text{instead of}\qquad
w_i = \frac{\pi(\theta_i)}{q_{t(i)}(\theta_i)},
\label{eq:balance-heuristic}
\end{equation}
the balance heuristic of \citet{veach1995}. The same denominator is used for particles from every stage. This allows earlier draws to contribute to the final estimate even after the proposal has changed. AMIS builds on adaptive importance sampling and population Monte Carlo \citep{cappe2004pmc,bugallo2017ais}. Consistency under regularity conditions was established by \citet{marin2019consistency}. We apply the same reweighting principle to a stream of individual draws. The proposal is updated after each arrival, and the denominator approximates the mixture of proposals used over the run. Recent work on parallel adaptive reweighting importance sampling also uses sample-centered Gaussian-mixture proposals and approximations to the cumulative proposal mixture \citep{liu2026paris}. Unlike our setting, that method assumes an evaluable target density and does not address asynchronous ABC simulations with heterogeneous runtimes.

\paragraph{Simulators with parameter-dependent cost.} Cell-based and agent-based models of tissues are a motivating example of simulators with parameter-dependent runtimes. These stochastic models represent individual cells either on a grid of sites (a lattice) or in continuous space. To infer model parameters from microscopy data, one can compare measured tissue features, such as cell counts and shape, with the corresponding features of simulated tissues. ABC is useful here because the likelihood is generally intractable \citep{kursawe2018abc,lambert2018agent}. The Cellular Potts model \citep{graner1992cpm,hirashima2017cpm} is the canonical lattice model of this class: each cell occupies a set of lattice sites, and cells grow, move and divide there. Because the number of simulated cells largely determines computational cost and depends on the inferred parameters, Cellular Potts runtimes can vary by orders of magnitude across the prior, making these simulations particularly susceptible to generation barriers (\S\ref{sec:experiments-suite}).

\paragraph{Barrier-free sampling elsewhere.} Parallel and off-barrier SMC has been studied for state-space models: the island particle model \citep{verge2015island}, asynchronous anytime particle methods \citep{paige2014cascade}, and parallel resampling \citep{murray2016anytime}. None of these targets ABC, and none combines barrier-free execution with AMIS-style reweighting against the mixture of proposals used. Look-ahead scheduling \citep{alamoudi2024lookahead} starts the next generation early on idle workers and so recovers part of the wait, but keeps the generation. 
Sequential neural simulation-based inference also commonly uses simulation rounds followed by retraining \citep{papamakarios2016snpe,greenberg2019apt,lueckmann2021benchmark}. More recently, dynamic SBI has removed fixed rounds from sequential neural posterior estimation by updating the simulation dataset and proposal during training \citep{lyu2026dynamicsbi}. Our method instead uses smooth-kernel ABC and reconstructs posterior weights directly from the completed simulations. Similar synchronization delays occur in distributed machine learning. In synchronous stochastic gradient descent, slow workers can delay an update \citep{dean2012distbelief,chen2016revisiting}; proposed remedies include backup workers and asynchronous updates \citep{recht2011hogwild}. Here the estimator is defined on whatever history has arrived, and the order of arrival matters only through the conditional distribution of each draw and how faithfully the posterior denominator represents it (\S\ref{sec:theory}).

\section{Method}\label{sec:method}
\begin{figure}[t]
\centering
\includegraphics[width=\linewidth]{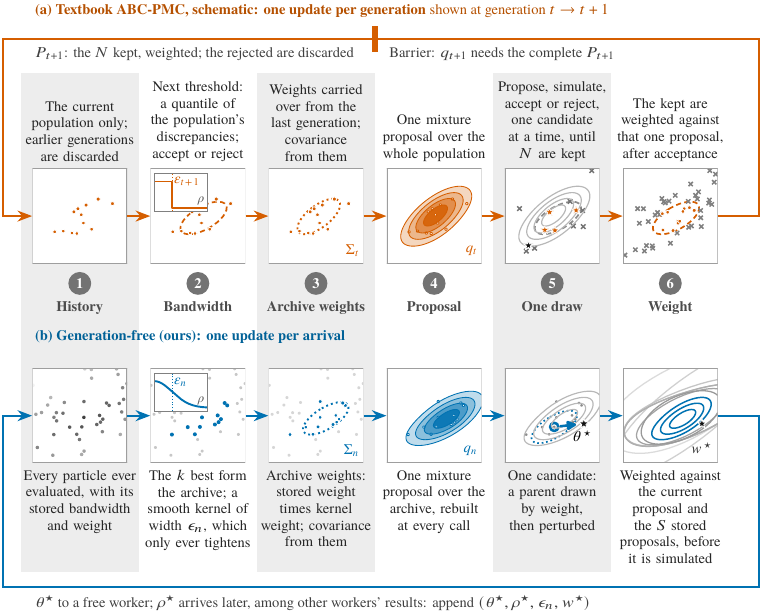}
\caption{One generation of ABC-PMC (a) beside one update of our sampler (b), step by step on a two-dimensional toy problem. ABC-PMC proposes a whole population from one mixture, simulates all of it, keeps the accepted particles and only then forms the next proposal, so a generation ends with its slowest simulation. Our sampler appends every completed evaluation to a history and rebuilds the archive, bandwidth and proposal from that history at every call. It draws one candidate for the free worker and weights that candidate against the current and stored proposals rather than against a single proposal. Row (a) is a schematic of textbook hard-threshold PMC, not the smooth-kernel pyABC baseline of \S\ref{sec:experiments-baselines}. Appendix~\ref{app:implementation} explains the symbols used in the figure.}
\label{fig:algorithm-strip}
\end{figure}

\subsection{Overview}\label{sec:method-overview}\label{sec:method-constraints}
Our sampler updates its proposal from completed evaluations and immediately assigns a new candidate to an available worker without waiting for the others. Figure~\ref{fig:algorithm-strip} contrasts this procedure (row b) with generation-based ABC-PMC (row a), which simulates a population before constructing its next proposal. Our method retains a growing history and uses the evaluations available when a worker requests another candidate (step 1). The $k$ particles with the smallest discrepancies form the \emph{archive}, and the bandwidth of the smooth ABC kernel is tightened as the archive improves (step 2). We construct a Gaussian-mixture proposal centered on the archive particles. Each mixture weight is proportional to the particle's stored parent weight multiplied by its kernel value at the current bandwidth (steps 3 and 4). One candidate is drawn from it and handed to the worker that has just become free (step 5). Since proposals change while simulations are running, we weight each candidate against a mixture of proposals rather than a single proposal. This applies the AMIS balance heuristic (\S\ref{sec:background}) to individual draws (step 6). The posterior estimate is not restricted to the last group of particles. After the run, we reweight all completed evaluations using a reconstructed mixture of sampling proposals.

Each worker constructs its next proposal from the history available to it, while the complete recorded history allows the posterior to be reconstructed during or after the run. Online parent selection uses an additional buffer of previous proposals (step 6), but this buffer is not needed for posterior reconstruction.

Weights serve different purposes during sampling and posterior estimation. During sampling, the stored \emph{parent weights} $w^\star$ and the current ABC kernel values determine the \emph{archive weights} $\tilde W_j$, which control the next proposal. For inference, we instead compute \emph{posterior weights} $W_{i,n}$ over the complete history after the run. The stored parent weights influence which candidates are proposed, but they are not the weights used to report the posterior. Sections~\ref{sec:method-history} to \ref{sec:method-weights} follow the figure's steps and \S\ref{sec:method-algorithm} assembles them.

\subsection{History and archive (steps 1 and 2)}\label{sec:method-history}
We retain all completed evaluations but use the $k$ best-fitting particles to construct the next proposal. This selected subset is the archive. Write the history at call $n$ as
\begin{equation}
\mathcal{H}_n = \bigl\{(\theta_i,\rho_i,\epsilon_i,w_i)\bigr\}_{i=1}^{n},
\label{eq:history}
\end{equation}
where $\theta_i$ is the parameter, $\rho_i$ the discrepancy, $\epsilon_i$ the bandwidth in force when $\theta_i$ was proposed, and $w_i$ its parent weight (\S\ref{sec:method-weights}). The archive $A_n$ is the set of the $k$ history members with the smallest discrepancies. Until the history holds $k$ members the sampler emits uniform prior draws (the \emph{bootstrap} phase), with $w_i=1$ and no bandwidth, so that the first archive can form. The history is append-only, and the archive, the bandwidth and the proposal below are recomputed from it at every call. The Propulate framework \citep{taubert2023propulate}, which hands the whole history to each call and takes one candidate back, provides this contract (Appendix~\ref{app:implementation}).

\subsection{Bandwidth (step 2)}\label{sec:method-tolerance}
Rather than accepting or rejecting each draw at a fixed threshold, we assign kernel weights as in \eqref{eq:abc-target}. As the archive improves, we reduce the bandwidth $\epsilon_n$ to give greater weight to better-fitting particles. Tightening it too quickly, however, can leave too few particles with appreciable weight. The scheduler therefore limits each decrease and waits for enough well-fitting evaluations before reducing the bandwidth further. The bandwidth never increases, and workers using different history prefixes follow a nonincreasing schedule. The bandwidth in force at call $n$ is the smaller of the running minimum of the bandwidths stored in the history and the value a bandwidth scheduler proposes from that history. The scheduler starts from a user-set $\epsilon_0$. At each search, the scheduler finds the largest bandwidth that retains at least $95\%$ of the archive's current kernel-weighted effective sample size. A single update can reduce the bandwidth by no more than a factor of two. The scheduler then holds the bandwidth until at least $2k$ evaluated particles have discrepancies below it, which brings $\epsilon_n$ close to the corresponding order statistic (settings in Appendix~\ref{app:implementation}). On an expensive simulator the run can end before the schedule has moved far from $\epsilon_0$, so $\epsilon_0$ has to be set from the scale of the prior-predictive discrepancies (\S\ref{sec:experiments-metrics}, \S\ref{sec:results-posterior}).

\subsection{Archive proposal (steps 3 to 5)}\label{sec:method-proposal}
At call $n$, we build the proposal by placing a Gaussian component at each archive member:
\begin{equation}
q_n(\theta) = \sum_{j\in A_n}\tilde W_j^{(n)}\, K_\Sigma(\theta-\theta_j),
\qquad
\tilde W_j^{(n)} \propto w_j \cdot K_{\epsilon_n}(\rho_j),
\label{eq:steady-state-proposal}
\end{equation}
The mixture weights $\tilde W_j^{(n)}$ are normalized to sum to one. Particles with larger weights are more likely to be selected as parents, and Gaussian perturbations explore nearby parameter values. The perturbation kernel $K_\Sigma$ has covariance $\Sigma_n = s\,\widehat{\mathrm{Cov}}_{\tilde W}(A_n)$, the weighted archive covariance scaled by $s$; each component is truncated to the prior box. Because the kernel is smooth, an archive member's contribution shrinks continuously as the bandwidth tightens instead of dropping out when a threshold crosses it, as it does in hard-threshold ABC-PMC. The numerical details are in Appendix~\ref{app:implementation}.

\subsection{Weights and the posterior estimate (step 6)}\label{sec:method-weights}\label{sec:method-amis}\label{sec:method-estimate}
Ordinary importance weights correct for the proposal used to generate a candidate. Here, however, the proposal changes with the stream of completed simulations, so the accumulated history contains draws from many proposals. We therefore use the AMIS idea of weighting against a mixture of proposals (\S\ref{sec:background}). During sampling, parent weights help determine which particles generate new candidates, whereas posterior estimation uses separate weights that account for the different proposals used throughout the run. A \emph{proposal snapshot} $q_s$ is the proposal \eqref{eq:steady-state-proposal} rebuilt from the history truncated at call $s$, at the bandwidth stored there. Any snapshot can be rebuilt from the history alone.

\emph{The parent weight.} The running sampler keeps the last $S$ snapshots it has taken, one every $k$ calls. A candidate $\theta^\star$ drawn from $q_n$ receives
\begin{equation}
w^\star = \frac{\pi(\theta^\star)}{\bar q^{\mathrm{on}}_n(\theta^\star)},
\qquad
\bar q^{\mathrm{on}}_n(\theta) = \frac{1}{1+|\mathcal{S}|}\Bigl[q_n(\theta) + \sum_{s\in\mathcal{S}} q_s(\theta)\Bigr],
\label{eq:amis-weight}
\end{equation}
with $\mathcal{S}$ the kept snapshots; $S=0$ recovers single-proposal weighting. This weight is stored with the particle as its $w_i$ and enters the archive weights of \eqref{eq:steady-state-proposal} and the scheduler's effective-sample-size search (\S\ref{sec:method-tolerance}). The kept snapshots are the only state the sampler carries between calls.

\emph{The posterior weights.} After the run, we reconstruct at most $S{+}1$ snapshots from $\mathcal{H}_n$ (two when $S=0$). Their history indices $\tau_1<\dots<\tau_m$ are spaced according to the numbers of draws represented by the snapshots. We weight each snapshot in proportion to that number, then reweight every evaluated particle against the resulting mixture. For a function $f$ of the parameter the estimate is
\begin{equation}
\widehat\pi_n(f) = \frac{\sum_{i=1}^n W_{i,n}\,f(\theta_i)}{\sum_{i=1}^n W_{i,n}},
\qquad
W_{i,n} = \frac{\pi(\theta_i)\,K_{\epsilon_n}(\rho_i)}{\bar q_n(\theta_i)},
\label{eq:posterior-estimator}
\end{equation}
with the denominator
\begin{equation}
\bar q_n = \lambda_n\,\pi + (1-\lambda_n)\sum_{s=1}^{m}\alpha_{s,n}\,q_{\tau_s},
\qquad \lambda_n=\max(\nu_n,\delta),\quad \alpha_{s,n}\ge0,\ \textstyle\sum_s\alpha_{s,n}=1.
\label{eq:snapshot-denominator}
\end{equation}
The denominator combines the reconstructed proposals with the prior. Each proposal snapshot receives a weight proportional to the number of draws it represents. The prior receives weight $\lambda_n$, given by the observed fraction $\nu_n$ of prior draws or a fixed minimum $\delta>0$, whichever is larger. This prior component bounds the importance weights by keeping the denominator bounded away from zero \citep{owen2000safe}. We use this finite snapshot mixture in place of the full mixture of sampling proposals. The accuracy of that approximation is the subject of \S\ref{sec:theory}. Computing $\widehat\pi_n$ is one $\mathcal{O}(nk)$ pass over the history, run after the timed budget and excluded from every throughput measurement (Appendix~\ref{app:implementation}).

\subsection{Update loop}\label{sec:method-algorithm}
Algorithm~\ref{alg:async-abc} is the update a worker runs whenever it needs a candidate; its lines follow the steps of Fig.~\ref{fig:algorithm-strip}. On $W$ workers each one runs the loop on its own copy of the history, which keeps growing as results arrive from every worker. No line waits for any other worker, and a worker that returns from a slow simulation finds a larger history than the one it left.
\begin{algorithm}
\caption{One update of the generation-free sampler, run by a worker on its copy of the history}\label{alg:async-abc}
\begin{algorithmic}
\Require History $\mathcal{H}_n$, archive size $k$, starting bandwidth $\epsilon_0$, kept snapshots $\mathcal{S}$
\State (1) \textbf{if} $|\mathcal{H}_n|<k$ \textbf{then return} a uniform prior draw with $w^\star=1$ \Comment{bootstrap}
\State (2) $A_n\gets$ the $k$ members of $\mathcal{H}_n$ with the smallest $\rho_i$; $\epsilon_n\gets$ the smaller of the stored running minimum and the scheduler's proposal
\State (3) $\tilde W^{(n)}_j\propto w_jK_{\epsilon_n}(\rho_j)$ over $A_n$; $\Sigma_n\gets s\,\widehat{\mathrm{Cov}}_{\tilde W}(A_n)$
\State (4) $q_n$ as in \eqref{eq:steady-state-proposal}; every $k$ calls, add $q_n$ to $\mathcal{S}$
\State (5) draw a parent $j\sim\tilde W^{(n)}$ and $\theta^\star\sim K_{\Sigma_n}(\cdot-\theta_j)$, retried into the prior box
\State (6) $w^\star$ as in \eqref{eq:amis-weight}
\State \Return $\theta^\star$ for simulation; $(\theta^\star,\rho^\star,\epsilon_n,w^\star)$ joins every worker's history when $\rho^\star$ arrives
\end{algorithmic}
\end{algorithm}
\section{Theoretical Analysis}\label{sec:theory}
Because the sampler changes its proposal throughout the run, importance weights must account for the different distributions from which parameters were drawn. We approximate their mixture using a bounded number of reconstructed proposals, reducing the cost of posterior estimation at the risk of introducing approximation error. We analyze this error for the estimator in \eqref{eq:posterior-estimator}, whose denominator $\bar q_n$ is a fixed-size mixture of reconstructed proposals. Write $\tilde q_i$ for the conditional density of the $i$-th candidate given the past (the filtration of Assumption~\ref{ass:filtration}). This conditional density captures the sampler's adaptive proposal mechanism. The theorem requires no further restriction on how these proposals evolve. The mixture the pooled sample was drawn from is
\begin{equation}
\bar q^\star_n:=\tfrac1n\textstyle\sum_{i\le n}\tilde q_i,
\label{eq:draw-mixture}
\end{equation}
and the \emph{fidelity ratio} of the run is $r_n:=\bar q^\star_n/\bar q_n$. This ratio compares the actual mixture of conditional sampling densities with the denominator used for reweighting. Where $\bar q_n$ is too small relative to $\bar q^\star_n$, the reweighted sample gives that region too much mass; where it is too large, the region receives too little. Theorem~\ref{thm:consistency} shows that the estimator tracks the smooth-ABC posterior at the realized bandwidth modified by $r_n$, with a vanishing tracking error under Assumptions~\ref{ass:kernel}--\ref{ass:filtration} and on $\{\epsilon_\infty>0\}$. When $r_n$ has a limit $r$, the estimate converges to the posterior tilted by $r$, so that $r\equiv1$ gives the posterior itself. The proof uses martingale arguments to handle proposals that depend on previously observed evaluations. Existing AMIS consistency results for recycling schemes \citep{marin2019consistency} require restrictions on adaptation that decouple the weighting stage. These restrictions do not hold for single-arrival updates, motivating a separate proof. Appendix~\ref{app:theory-more} gives the decomposition and partial measurement of $r$; the proofs are in Appendix~\ref{app:theory}. The bandwidth sequence is nonincreasing, so its limit $\epsilon_\infty:=\lim_n\epsilon_n$ exists, although it may equal zero. Our asymptotic statements apply on the event $\{\epsilon_\infty>0\}$. We return to the practical implications of this condition at the end of the section. $L_\epsilon(\theta)=\mathbb{E}[K_\epsilon(\rho)\mid\theta]$ is the smooth ABC likelihood of \eqref{eq:abc-target}.
\paragraph{Assumptions.}
\begin{assumption}[kernel]\label{ass:kernel} $K_\epsilon:[0,\infty)\to[0,1]$ is measurable with $K_\epsilon(0)=1$ (unit \emph{peak}, not unit integral, as is standard for a probabilistic-acceptance ABC kernel), and for every $a\in(0,\epsilon_0]$ there is $L_K(a)<\infty$ with $|K_\epsilon(\rho)-K_{\epsilon'}(\rho)|\le L_K(a)|\epsilon-\epsilon'|$ for all $\rho\ge0$ and $\epsilon,\epsilon'\in[a,\epsilon_0]$. The limiting target is non-degenerate: $Z_{\epsilon_\infty}=\int_\Theta\pi L_{\epsilon_\infty}>0$.\end{assumption}
\begin{assumption}[denominator class]\label{ass:class} $\Theta\subset\mathbb{R}^d$ is a non-degenerate axis-aligned box and the prior density satisfies $0<\pi_-\le\pi\le\pi_+<\infty$ on it. Each $q_{\tau_s}$ in \eqref{eq:snapshot-denominator} is a positive function on $\Theta$ of truncated-Gaussian-mixture form: at most $k$ components, means in $\Theta$, weights in the simplex, covariance $\Sigma$ with $\sigma_-^2I\preceq\Sigma\preceq\sigma_+^2I$ for fixed $0<\sigma_-\le\sigma_+<\infty$, and per-component normalizing constants evaluated as the implementation evaluates them.\end{assumption}
\begin{assumption}[prior floor in the denominator]\label{ass:defensive} $\delta\in(0,1/2]$ is fixed and $\lambda_n\ge\delta$ in \eqref{eq:snapshot-denominator}.\end{assumption}
\begin{assumption}[execution model]\label{ass:filtration} Candidates are indexed by proposal time. \emph{(B1)} Given its parameter $\theta$, a simulator call returns a discrepancy and a wall-clock duration, $(\rho,D)\sim P(\cdot\mid\theta)$, and the pairs of distinct calls are conditionally independent given their parameters. \emph{(B2)} The sampler's own randomness (parent choice, perturbation, rejection retries, redraws, bootstrap draws) is drawn fresh at each call from a stream independent of the simulator's. \emph{(B3)} Earlier proposal times, durations and worker identities, together with an exogenous source of randomness (message delivery), determine which worker proposes next, when it proposes, and which earlier results it has received by then. The exogenous source is independent of the streams in (B1) and (B2). Lemma~\ref{lem:adapted} derives from this model the filtration $(\mathcal{F}_i)$ used in the proofs, the record, in proposal order, of every candidate launched so far with its result, duration and worker. Under that filtration, $\theta_i\mid\mathcal{F}_{i-1}\sim\tilde q_i$, with $\tilde q_i$ the proposal density the worker emits, and $\rho_i\mid\mathcal{F}_{i-1}\vee\sigma(\theta_i)\sim p(\cdot\mid\theta_i)$. Under sequential execution (B3) is empty.\end{assumption}
\begin{assumption}[fidelity]\label{ass:stab} The denominator tracks the draw mixture up to a limiting ratio: for some $r$ with $0<r_-\le r\le r_+<\infty$,
\begin{equation}
\bigl\|\bar q^\star_n/\bar q_n-r\bigr\|_\infty\longrightarrow0\qquad\text{a.s.}
\label{eq:fidelity}
\end{equation}
\end{assumption}

Assumption~\ref{ass:kernel} holds for the Gaussian kernel with $L_K(a)=2/(ea)$ and Epanechnikov with $2/a$; it excludes the hard indicator, which is not Lipschitz in $\epsilon$. This is one reason the method uses smooth kernels. Assumptions~\ref{ass:class} and \ref{ass:defensive} constrain only the $m$-snapshot denominator $\bar q_n$, which the implementation builds to that exact recipe. Assumption~\ref{ass:filtration} is a property of the execution, not of the proposal path; Table~\ref{tab:assumption-status} records it for the implementation. Assumption~\ref{ass:stab} is the only place the sampler's proposal path enters; Theorem~\ref{thm:consistency} does not need it.

\begin{theorem}[tracking at the realized bandwidth]\label{thm:consistency}
Let Assumptions~\ref{ass:kernel}--\ref{ass:filtration} hold and write $r_n:=\bar q^\star_n/\bar q_n$. Then for every bounded measurable $h:\Theta\to\mathbb{R}$, almost surely on $\{\epsilon_\infty>0\}$,
\begin{equation}
\frac1n\sum_{i\le n}W_{i,n}\,h(\theta_i)\;-\;\int_\Theta \pi\,h\,L_{\epsilon_n}\,r_n\;\longrightarrow\;0 .
\label{eq:tracking}
\end{equation}
Consequently, if in addition $\liminf_n\tfrac1n\sum_{i\le n}W_{i,n}>0$, then for every bounded measurable $f$,
\begin{equation}
\widehat\pi_n(f)\;-\;\pi^{r_n}_{\epsilon_n}(f)\;\longrightarrow\;0,
\qquad
\pi^{\varrho}_{\epsilon}(f):=\frac{\int_\Theta f\,\pi\,L_{\epsilon}\,\varrho}{\int_\Theta \pi\,L_{\epsilon}\,\varrho}.
\label{eq:tracking-ratio}
\end{equation}
\end{theorem}

Theorem~\ref{thm:consistency} distinguishes two sources of error: random variation in the weighted estimate and the approximation of its denominator. The first vanishes under the stated assumptions and the remaining difference is described by $r_n$. The normalized conclusion also requires that the average unnormalized weight does not collapse, in addition to a limiting bandwidth above zero. These conditions can be checked against a recorded run. Passing to a limit needs the ratio to settle:

\begin{corollary}[tilted limit]\label{cor:tilt}
Let Assumptions~\ref{ass:kernel}--\ref{ass:stab} hold. Then for every bounded measurable $f$, almost surely on $\{\epsilon_\infty>0\}$,
\begin{equation}
\widehat\pi_n(f)\;\longrightarrow\;\pi^{r}_{\epsilon_\infty}(f)=\frac{\int_\Theta f\,\pi\,L_{\epsilon_\infty}\,r}{\int_\Theta \pi\,L_{\epsilon_\infty}\,r}.
\label{eq:tilted-target}
\end{equation}
\end{corollary}

\begin{corollary}[exact fidelity]\label{cor:exact}
If moreover $r\equiv1$, that is, if the denominator is asymptotically the density the sample was drawn from, then $\widehat\pi_n(f)\to\pi_{\epsilon_\infty}(f)$, the smooth-ABC posterior at the limiting bandwidth.
\end{corollary}

\paragraph{What $r$ is, and how much of it is measured.} Four implementation choices contribute to $r$: the prior floor $\delta$ in the denominator, the $m$ reconstructed snapshots standing in for the $n$ proposals actually used, the approximate normalization of the truncated mixture components, and the sampler's departures from drawing exactly from $q_n$ (Appendix~\ref{sec:theory-r}). The first two can be estimated from a stored run by recomputing the denominator at a much larger $m$ and at the observed prior share: the estimate lies $0.2$--$0.5\%$ in total variation from that reference on the five Cellular Potts production runs and $0.26\%$ on a one-dimensional Gaussian history. We bound the other two effects rather than measuring them. Because two contributions to $r$ are bounded rather than measured, the empirical assessment remains partial. The limiting results additionally require $r$ to converge to a finite, strictly positive function (Table~\ref{tab:assumption-status}). The experiments use $S=20$, so $m\le21$, and the floor is $\delta=0.5/(m{+}1)$ (Appendix~\ref{app:implementation}).

\paragraph{Which runs the limits describe.} Our convergence results require the bandwidth used for posterior estimation to remain positive as the number of evaluations increases. A minimum bandwidth ensures this condition, though it does not itself guarantee that the reconstructed proposal mixture is accurate. The condition is $\{\epsilon_\infty>0\}$ for the reporting bandwidth $\epsilon_n$ in \eqref{eq:posterior-estimator}. The bandwidths used to construct proposals enter the proofs only through $\tilde q_i$, on which Theorem~\ref{thm:consistency} places no condition. We recommend setting the minimum to $10^{-4}\epsilon_0$ (\texttt{min\_tol}, Appendix~\ref{app:implementation}). The floor acts only by clipping a bandwidth that would fall below it, so a run whose schedule stays above the floor follows exactly the path it would have followed with the floor set. That holds for every g-and-k and Cellular Potts run whose posterior we report, including a replayed sample of the g-and-k calibration trials: their final bandwidths lie more than ten times above $10^{-4}\epsilon_0$ (Table~\ref{tab:assumption-status}), so the limits describe these runs as they were executed.
The Gaussian-mean and Lotka--Volterra benchmarks, inexpensive simulators that run for millions of evaluations, tighten to below $10^{-5}\epsilon_0$ and are not covered at that floor. These conditions do not imply that a finite simulation budget has reached the asymptotic regime. No reported run reached its floor, so each is still an estimate whose effective support is of the order of the archive size $k$: the scheduler holds the bandwidth until $2k$ evaluated particles lie inside it (\S\ref{sec:method-tolerance}), only $O(k)$ particles carry appreciable kernel mass, and the effective-sample-size ceiling of \S\ref{sec:results-posterior} follows from this. Once the floor is reached, the average weight converges to a positive constant and the effective sample size grows linearly in $n$; the limits describe that continuation, and the measurements above and the calibration checks of \S\ref{sec:results-posterior} quantify the error at the budgets used.

\section{Experimental Design}\label{sec:experiments}

\subsection{Benchmarks}\label{sec:experiments-suite}
We use separate benchmarks to assess posterior accuracy, synchronization overhead, and performance on a cell-based simulator. Table~\ref{tab:benchmarks} lists the five configurations, with $d$ the number of inferred parameters and CV the coefficient of variation of the runtime. The one-dimensional Gaussian mean provides an analytic posterior, while the four-dimensional g-and-k model \citep{rayner2002gk} has a reference posterior. Their short simulation times also make them useful for measuring per-arrival coordination overhead. The two studies with injected runtime heterogeneity run on the Gaussian mean: a persistently slow worker and a lognormal multiplier of spread $\sigma$ applied to every evaluation's runtime. Lotka--Volterra is a stochastic predator--prey simulator whose cost varies with the trajectory and whose parameters send $98\%$ of simulations to extinction (Appendix~\ref{app:protocol}). It is used for strong-scaling measurements, which increase the worker count at a fixed problem size, and to locate the cost boundary: the per-simulation cost above which barrier removal saves more time than per-arrival coordination consumes. It is not scored for posterior quality.

In the Cellular Potts benchmark, runtime variation comes from the model itself rather than from an imposed slowdown. How fast cells divide and how large they grow determine how many cells the lattice holds, and the cell count fixes both the summaries and the runtime (correlation $0.99$ between runtime and cell count). We run the \texttt{cellsInSilico} implementation \citep{berghoff2020cis} of the three-dimensional Cellular Potts model \citep{graner1992cpm} at two lattice sizes. We infer cell division rate and target cell volume using the final cell count and cluster radius as summary statistics. Each evaluation averages four simulations, since individual realizations are too noisy to reliably rank nearby parameter values. The simulations take seconds to minutes rather than the hours of a production study, and the observations are synthetic with a known truth, so the benchmark scores both the systems claims and posterior recovery. The $50^3$ runs of the synchronized variant (next subsection) and of the scaling study used an earlier parameterization of the same simulator (division rate and cell motility; one simulation per evaluation); only its runtime distribution enters those results. Appendix~\ref{app:cpm} gives the screening experiment that fixed the parameters and summaries, the prior ranges, and both configurations.
\begin{table}[t]
\caption{Benchmark suite. Cost is the mean wall-clock time of one simulation in the asynchronous runs (an evaluation on Cellular Potts is four replicate simulations); runtime spread is the in-run coefficient of variation of that time, or how it is injected. Budgets are per replicate; five replicates throughout except three at $80^3$.}\label{tab:benchmarks}
\centering
\footnotesize
\setlength{\tabcolsep}{3pt}
\begin{tabular}{@{}>{\raggedright\arraybackslash}p{0.13\textwidth}c>{\raggedright\arraybackslash}p{0.085\textwidth}>{\raggedright\arraybackslash}p{0.20\textwidth}>{\raggedright\arraybackslash}p{0.15\textwidth}>{\raggedright\arraybackslash}p{0.255\textwidth}@{}}
\toprule
Benchmark & $d$ & cost/sim & runtime spread & budget, workers & role \\
\midrule
Gaussian mean & 1 & $0.1$\,ms & injected: straggler $0$--$20\times$; lognormal $\sigma\le2$ & $300$\,s, $48$ ($16$ in the straggler study) & calibration; barrier cost under a known runtime distribution; the near-instantaneous end of the cost boundary \\
g-and-k & 4 & $2$\,ms & none & $600$\,s, $48$ & 4-D calibration and recovery against a reference posterior \\
Lotka--Volterra & 4 & $4$\,ms & trajectory-dependent; $98\%$ extinct & $600$\,s, $48$; scaling to $288$ & strong scaling; the cost boundary \\
Cellular Potts $50^3$ & 2 & $13$\,s & CV $0.05$ & $3600$\,s, $48$; scaling to $384$ & posterior recovery against a known truth; barrier cost on a real simulator \\
Cellular Potts $80^3$ & 2 & $170$\,s & CV $0.13$--$0.23$ & $3600$\,s, $48$ & barrier cost at realistic cost and heterogeneity \\
\bottomrule
\end{tabular}
\end{table}

\subsection{Baselines and instruments}\label{sec:experiments-baselines}\label{sec:implementation}
We use three comparators to answer different questions. The \emph{synchronized variant} isolates the cost of adding a barrier to our sampler, \emph{pyABC} provides a practical baseline for overall performance, and \emph{rejection ABC} tests the benefit of an adaptive proposal over prior sampling. For the practical baseline, we modify pyABC \citep{klinger2018pyabc,schalte2022pyabc} to use the same smooth-kernel acceptance rule and bandwidth schedule as the asynchronous sampler (Appendix~\ref{app:implementation}); it is the comparator for C2 and C3. With a population of $100$ and $47$ simulating workers (one process only dispatches), pyABC requires at least three rounds per generation. The last round is $13\%$ full, so pyABC idles even on a homogeneous simulator. The scaling study gives it a population equal to the worker count to alleviate this issue. Our synchronized variant uses the same propagator (Propulate's operator that maps the current history to the next candidate) but imposes a collective barrier before each proposal. It therefore differs in synchronization, not in its proposal-update mechanism. Histories are not truncated to generation boundaries, so each worker proposes from the history it has received at the barrier. Comparing against this synchronized variant lets us attribute differences in worker idle time to synchronization rather than to a different sampling algorithm, and it is the comparator for C1. Because a collective barrier needs identical call counts on every worker, the synchronized variant runs to a fixed simulation count and the asynchronous sampler to a fixed wall clock. We therefore compare them as rates. For the posterior comparison the asynchronous sampler is also re-run at exactly the synchronized variant's simulation count and seeds. \emph{Rejection ABC} draws candidates from the prior and retains the $k$ lowest-discrepancy samples. It receives the same simulation budget and uses all $48$ workers, rather than one worker per replicate. It is the reference for the per-simulation gain of an adaptive proposal. All three run on the JUWELS cluster over the Message Passing Interface (MPI; Appendix~\ref{app:protocol}) and are samplers of the same estimator class.

\subsection{Metrics}\label{sec:experiments-metrics}
\paragraph{Tolerance at a matched budget.} To compare samplers at the same simulation budget, we use the $k$-th smallest discrepancy among a method's first $n$ completed evaluations, denoted $\epsilon_{(k)}(n)$. It is the threshold at which the best $k$ draws would be accepted. It needs no reference posterior, it does not depend on the bandwidth a schedule happened to report at, and it is what rejection ABC reports at by construction. The \emph{per-simulation efficiency} of the asynchronous sampler over pyABC is the ratio of their $\epsilon_{(k)}$ at the same $n$, above one when the asynchronous tolerance is the tighter; the \emph{equal-wall-clock} ratio compares their tolerances at the end of the same wall-clock budget and therefore also reflects differences in throughput. The slope of $\epsilon_{(k)}(n)$ in $\log n$ shows whether more simulations keep converting into a tighter tolerance. We use $k=100$ throughout, over completed simulations only, ordered by completion time (Appendix~\ref{app:protocol}).

\paragraph{Throughput, utilization, and the predicted barrier cost.} Throughput measures completed simulations per second of wall clock, whereas utilization measures the fraction of worker time spent inside the simulator. We use durations recorded during the asynchronous run to predict the throughput lost to synchronization and test this prediction against the synchronized variant (C1). Let $D_i$ be those per-evaluation durations (recorded simulation time plus the measured per-evaluation coordination overhead). A barrierized generation of $W$ evaluations takes $\mathbb{E}[\max_{i\le W} D_i]$, so the synchronous throughput is predicted as $T_{\mathrm{sync}}=W/\mathbb{E}[\max_{i\le W}D_i]$. The expectation is taken by resampling $W$ durations from the asynchronous record, or from the injected runtime distribution where the heterogeneity is injected. The predicted cost is $T_{\mathrm{async}}/T_{\mathrm{sync}}$ with $T_{\mathrm{async}}$ the measured asynchronous throughput, and it is compared with the ratio measured against the synchronized variant. For identically distributed, uncensored durations, $T_{\mathrm{async}}=W/\mathbb{E}[D]$, and the predicted synchronization penalty reduces to $\mathbb{E}[\max_{i\le W}D_i]/\mathbb{E}[D]$. With a persistently slow worker, or when a wall-clock budget censors the long durations, the two differ, and we use the general form. The prediction has a domain: the simulation time must dominate the latency of the collective barrier itself, which the model does not contain, and the timing sample must be long enough to cover the tail of the runtime distribution. Configurations outside this domain are marked as such in \S\ref{sec:results-barrier}.

\paragraph{Posterior quality.} We assess calibration using simulation-based calibration \citep{talts2018sbc}. In each of $1000$ trials, we draw a parameter from the prior and simulate a dataset at that value. We then check the empirical coverage of central posterior intervals at nominal levels $0.50$, $0.80$, $0.90$, and $0.95$. We also compute the rank of the true parameter among posterior draws; for a calibrated posterior, these ranks are uniformly distributed. The effective sample size (ESS) of a weighted posterior is $(\sum_i W_i)^2/\sum_i W_i^2$. Where a reference posterior exists we report the Wasserstein distance $W_1$ between each method's own posterior and that reference (in more than one dimension the sliced $W_1$ over random projections \citep{bonneel2015sliced}, on the parameters' natural scale). The reference is the analytic posterior for the Gaussian mean; for g-and-k it is the posterior given the same octile summaries the discrepancy is built from, obtained from their asymptotic distribution (Appendix~\ref{app:protocol}). The asynchronous estimator is replayed over the prefix of the history completed by each checkpoint. pyABC is scored on its last completed generation, with pyABC's population weights where the history was kept and uniform weights otherwise, and rejection ABC on its $k$ best draws. On Cellular Potts, whose truth is known but whose posterior is not, we report each parameter's \emph{contraction} $1-\mathrm{sd}/\mathrm{sd}_{\mathrm{prior}}$ in prior-normalized units, whether the central $90\%$ interval covers the truth, and the ESS of the posterior weights.

\paragraph{Starting bandwidth.} On an expensive simulator the bandwidth schedule may not move far from $\epsilon_0$ within the budget (\S\ref{sec:method-tolerance}), so $\epsilon_0$ is set before the run from the discrepancies of a few thousand prior-predictive draws, at about a fifth of their median. \S\ref{sec:results-posterior} measures what a starting value far above that costs.

\section{Results}\label{sec:results}
Unless noted, results are medians over five replicates with the replicate range shown, \emph{asynchronous} is the proposed sampler, \emph{pyABC} the baseline with the same kernel and schedule, and \emph{synchronized variant} our sampler with a collective barrier inserted (\S\ref{sec:experiments-baselines}).

\subsection{Predicting the cost of synchronization barriers (C1)}\label{sec:results-barrier}

\paragraph{Predicted and measured synchronization costs.} Figure~\ref{fig:barrier} compares the runtime-based prediction with measured synchronization costs for every configuration on which we ran the synchronized variant: a persistently slow worker at five slowdowns, lognormal runtime heterogeneity at five spreads, and the Cellular Potts simulator at four worker counts and two lattice sizes. For each, the barrier's cost as predicted from the asynchronous run's timing alone (\S\ref{sec:experiments-metrics}) is plotted against the cost measured with the synchronized variant. Over the twelve configurations inside the predictor's domain the predicted-to-measured ratio runs from $0.84$ to $1.19$ with median $0.996$, across measured costs from $1.2\times$ to $400\times$ (Tables~\ref{tab:twin}, \ref{tab:twin-hetero} and \ref{tab:twin-cpm} in Appendix~\ref{app:moved}). The largest prediction errors are $16\%$ and $19\%$, both at the widest runtime spreads. In these configurations, the synchronized variant completes only ten generations per replicate, while the $60$\,s budget censors the longest asynchronous evaluations. Three configurations lie outside the domain and are drawn open. The first two are the $0\times$ and $1\times$ slowdowns of the straggler study on the cost-free simulator, where the collective's own latency, about $0.3$\,s per barrier, dominates a $4$\,ms workload. At $384$ Cellular Potts workers, contention among the workers gives the asynchronous timing sample a tail of long durations that the synchronized variant's generations did not show.

\begin{figure}[t]
\centering
\includegraphics[width=0.55\linewidth]{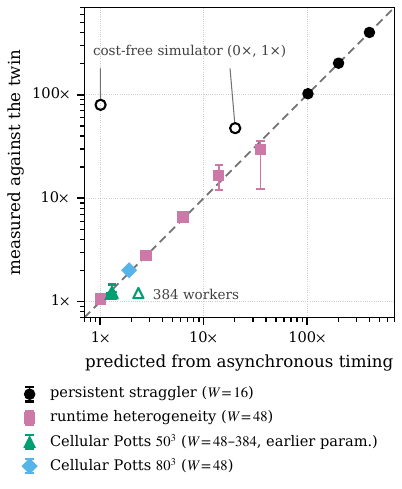}
\caption{The barrier's cost predicted from the asynchronous run's timing alone (\S\ref{sec:experiments-metrics}) against the cost measured with the synchronized variant, on every configuration the synchronized variant was run on: a persistent straggler at five slowdowns, lognormal runtime heterogeneity at five spreads, and Cellular Potts at $50^3$ ($48$--$384$ workers, earlier parameterization, utilization ratio) and $80^3$ ($48$ workers, against pyABC). Points show replicate medians and ranges. The dashed line indicates equal predicted and measured costs; open markers identify the three configurations outside the predictor's domain.}\label{fig:barrier}
\end{figure}

\paragraph{Separating synchronization overhead from other effects.} On Cellular Potts at $50^3$ the synchronized variant's throughput was $1.9$--$2.8\times$ below the asynchronous sampler's, but only part of that is the barrier. Since we cannot attribute the synchronized variant's longer simulation times to the barrier, we compare worker utilization to isolate the synchronization penalty. The resulting $1.2\times$ ratio agrees with the prediction at every worker count (Table~\ref{tab:twin-cpm}). The barrier's cost on this workload is therefore $1.2\times$, as its runtime spread predicts, and it rises with the lattice: at $80^3$, where an evaluation costs $170$\,s and the runtime coefficient of variation is $0.23$, the predicted cost is $1.9\times$ and the gap measured against pyABC $2.1\times$. The same comparison against pyABC at $50^3$ gives $2.3$--$2.4\times$, of which only $1.2\times$ is the barrier. The rest is pyABC's fixed per-generation overhead, a large fraction of a $13$\,s evaluation and negligible against a $170$\,s one. The asynchronous sampler itself scales almost linearly on this workload, at $97\%$ parallel efficiency (throughput relative to perfect linear scaling) from one to eight nodes (Fig.~\ref{fig:scaling}, right).

\paragraph{Against pyABC.} The same experiments score pyABC (Fig.~\ref{fig:barrier-pyabc}, Appendix~\ref{app:moved}). Under the persistent straggler its throughput falls more than a hundredfold while the asynchronous sampler's falls by less than half, because every generation waits for the slow worker. Under runtime heterogeneity its workers idle $51$--$79\%$ of their time against $1$--$27\%$ for ours, and its throughput is $2.1$--$4.8\times$ below ours. That is less than the synchronized variant loses, because pyABC closes a generation on the first $100$ finishers and discards the late results of simulations it has already run, so the comparison a practitioner faces understates the barrier's full cost. On Cellular Potts pyABC runs at $43$--$49\%$ utilization against $98$--$99.8\%$, and the asynchronous sampler completes $2.1$--$2.4\times$ its simulations in the same budget.

\begin{figure}[t]
\centering
\includegraphics[width=\linewidth]{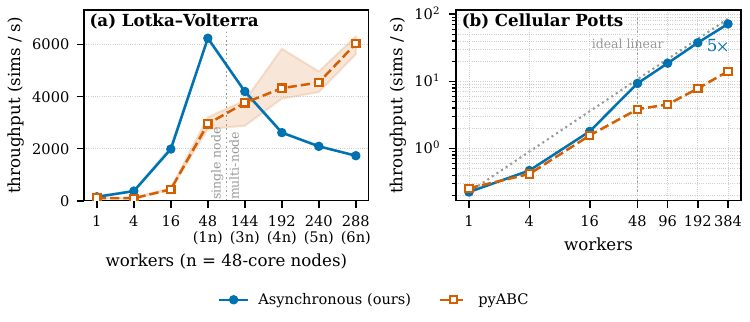}
\caption{Strong scaling at a fixed wall-clock budget against pyABC with population $\max(100,W)$. \textbf{Left:} Lotka--Volterra ($4$\,ms; $180$\,s; medians and inter-quartile ranges over five replicates). \textbf{Right:} Cellular Potts $50^3$ (earlier parameterization, $5.1$\,s; $1800$\,s; log--log), with linear scaling dotted.}\label{fig:scaling}
\end{figure}

\subsection{Throughput converts into a tighter tolerance (C2)}\label{sec:results-converts}

We first examine whether the additional simulations reduce the empirical tolerance. Table~\ref{tab:matched-eps} reports $\epsilon_{(100)}(n)$ for each benchmark (curves in Fig.~\ref{fig:eps-curves}, Appendix~\ref{app:moved}). At pyABC's own simulation count the asynchronous sampler is the more efficient per simulation on three of the four benchmarks; the one on which it is worse is the one-dimensional analytic target. These two comparisons address different questions: how well a method uses each simulation and how far it progresses within a fixed time. At equal wall clock, the asynchronous sampler reaches a tighter tolerance on Lotka--Volterra and a $4.1\times$ tighter tolerance on Cellular Potts. On the two benchmarks with near-instantaneous simulators, the asynchronous and pyABC tolerance curves have similar slopes, suggesting an approximately constant per-simulation efficiency ratio. On Cellular Potts, both adaptive methods reduce tolerance faster than the $n^{-1}$ trend of rejection ABC in this noise-dominated regime. Their advantage therefore increases with the simulation budget.
\begin{table}[t]
\caption{Tolerance at a matched budget on every benchmark ($k=100$ order statistic of each method's own discrepancies; five replicates; medians, with the replicate range of each ratio in parentheses; completed simulations only). Per-simulation: $\epsilon$ of pyABC over the asynchronous sampler at pyABC's own simulation count; equal wall clock: the same ratio at the end of both runs. Bold marks the larger simulation count.}\label{tab:matched-eps}
\centering
\footnotesize
\setlength{\tabcolsep}{4pt}
\begin{tabular}{@{}lrrccc@{}}
\toprule
Benchmark & cost/sim & sims (async / sync) & throughput & per-simulation & equal wall clock \\
\midrule
Gaussian mean   & $0.1$\,ms & $1.31$M / \textbf{{\boldmath$4.61$}M} & $0.28\times$ & $0.78\times$ ($0.77$--$1.12$) & $0.24\times$ ($0.21$--$0.28$) \\
g-and-k         & $2$\,ms   & $1.75$M / \textbf{{\boldmath$3.81$}M} & $0.46\times$ & $1.06\times$ ($1.02$--$1.15$) & $0.88\times$ ($0.85$--$0.92$) \\
Lotka--Volterra & $4$\,ms   & $1.51$M / \textbf{{\boldmath$1.56$}M} & $0.96\times$ & $1.30\times$ ($1.08$--$1.51$) & $1.30\times$ ($1.02$--$1.43$) \\
Cellular Potts  & $13$\,s   & {\boldmath$12{,}900$} / $5{,}400$ & $2.39\times$ & $1.6\times$ ($1.14$--$1.94$) & $4.09\times$ ($3.1$--$6.6$) \\
\bottomrule
\end{tabular}
\end{table}

\paragraph{Cellular Potts, fixed production run.} Table~\ref{tab:cpm-production} reports the production run, where both higher throughput and greater per-simulation efficiency contribute to the tighter tolerance at equal wall clock. In the same hour the asynchronous sampler completes $2.4\times$ the simulations of pyABC, of which \S\ref{sec:results-barrier} attributes $1.2\times$ to the barrier at this lattice size and most of the rest to pyABC's per-generation overhead. It is $1.6\times$ more efficient per simulation, read at pyABC's count of $5{,}400$ where pyABC's curve is the steeper (Fig.~\ref{fig:eps-curves}d), and ends $4.1\times$ tighter at equal wall clock. A rejection sampler given the same budget and the whole machine ends $39\times$ looser. At $80^3$ the asynchronous sampler completes $2.1\times$ the simulations and is $2.1\times$ more efficient per simulation.
\subsection{Where the net advantage becomes positive (C3)}\label{sec:results-boundary}

Each arrival incurs a roughly fixed coordination cost: rebuilding the $\mathcal{O}(k)$ proposal, weighting the arrival against the kept snapshots, and exchanging it with every other worker. This cost becomes less important as the simulation cost increases. Figure~\ref{fig:crossover} plots the throughput and per-simulation ratios of Table~\ref{tab:matched-eps} against the cost of one simulation. Both ratios increase with simulation cost. Across benchmarks, throughput improves relative to pyABC as the fixed per-arrival overhead becomes smaller compared with simulation time. Changes in per-simulation efficiency are harder to compare, because the inference problems differ in dimension, summaries, and noise. We therefore base the boundary on throughput: parity is reached at about $4$\,ms per simulation, and the net advantage at equal wall clock turns positive between $2$ and $4$\,ms. This is an empirical crossover for our worker placement and time budget, not a universal threshold. The benchmarks with inexpensive simulators ran $48$ workers spread over four nodes (Appendix~\ref{app:protocol}), so every arrival crosses the network. On one packed node, the same Lotka--Volterra configuration gives a throughput ratio of $2.1$ instead of $0.96$ (Table~\ref{tab:scaling}). Below the boundary, the asynchronous sampler's coordination costs exceed the time saved by removing the barrier: on the $0.1$\,ms Gaussian mean pyABC completes $3.5\times$ more simulations and ends with a $4\times$ tighter tolerance at equal wall clock.

When evaluations are low-cost, coordination also limits scaling to larger worker counts. On Lotka--Volterra the exchange of every arrival with every worker grows with the worker count, so a worker is inside the simulator $25\%$ of its time at $48$ workers and $2\%$ at six nodes; throughput peaks on one node, and pyABC with a population equal to the worker count overtakes beyond three nodes (Fig.~\ref{fig:scaling}, left). At $13$\,s per evaluation the same coordination is negligible, and the method tracks linear scaling to eight nodes (Fig.~\ref{fig:scaling}, right).
\begin{figure}[t]
\centering
\includegraphics[width=0.55\linewidth]{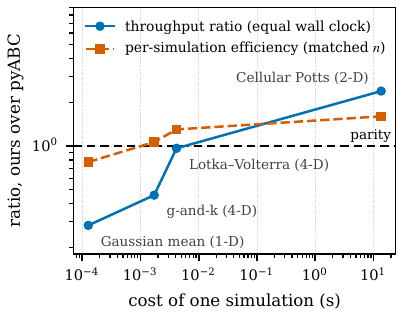}
\caption{Throughput ratio and per-simulation efficiency of the asynchronous sampler over pyABC against the cost of one simulation, one point per benchmark (Table~\ref{tab:matched-eps}).}\label{fig:crossover}
\end{figure}

\subsection{The posterior: calibration, recovery, and what limits it (C4)}\label{sec:results-posterior}

\paragraph{Calibration.} Table~\ref{tab:sbc} reports simulation-based calibration of the full-history estimator ($1000$ trials, matched simulation budget), showing near-nominal coverage for the one-dimensional Gaussian and bimodal targets but conservative coverage for two parameters of the four-dimensional g-and-k model. On the one-dimensional Gaussian mean it tracks the nominal coverage at every level, while pyABC under-covers at every level (its nominal $0.90$ interval covers $0.75$); the difference exceeds the Monte Carlo error of ${\pm}0.01$. On a bimodal target ($y_i\sim\mathcal N(\theta^2,\sigma^2)$, the case a top-$k$ archive is most at risk of collapsing onto one mode) the archive retains both modes in $86\%$ of trials and coverage is nominal. On the four-dimensional g-and-k model the estimator is \emph{conservative} on $A$ and $g$ and near nominal on $B$ and $k$, with $B$'s rank histogram the only uniform one at the $1\%$ level. pyABC is near nominal on $A$ and $g$ and under-covers $B$ and $k$, with no uniform rank histogram. The top-$k$ archive under-covers ($0.85/0.89$ at nominal levels $0.90/0.95$), whereas the full-history estimator over-covers. Coverage increases monotonically through nominal as the reporting support grows from the archive to the full history, pointing to reporting-support size rather than importance weighting as the main source of miscalibration. Doubling the simulation budget leaves the coverage unchanged. A sweep over the archive size and the number of kept snapshots on a $d$-dimensional Gaussian target (Appendix~\ref{app:moved}) finds a miscalibrated region at $d=8$, $k=50$. The same sweep of the archive sizes indicates $k=100$ is the best calibrated.

\begin{table}[t]
\caption{Simulation-based calibration of the posterior estimate ($1000$ trials; empirical coverage of central intervals at the nominal levels; Monte Carlo standard error ${\pm}0.007$--$0.016$). Nominal is best, and bold marks the method closer to it.}\label{tab:sbc}
\centering
\begin{tabular}{@{}llcccc@{}}
\toprule
Target & Method & 0.50 & 0.80 & 0.90 & 0.95 \\
\midrule
Gaussian mean (1-D)     & asynchronous & \textbf{0.51} & \textbf{0.80} & \textbf{0.90} & \textbf{0.94} \\
                        & pyABC & 0.35 & 0.63 & 0.75 & 0.83 \\
Bimodal $\mathcal N(\theta^2,\sigma^2)$ (1-D) & asynchronous & \textbf{0.50} & \textbf{0.82} & 0.91 & 0.96 \\
                        & pyABC & 0.45 & \textbf{0.82} & \textbf{0.90} & \textbf{0.95} \\
g-and-k (4-D), $A$      & asynchronous, full history & 0.62 & 0.91 & 0.96 & 0.99 \\
                        & pyABC & \textbf{0.54} & \textbf{0.82} & \textbf{0.91} & \textbf{0.95} \\
g-and-k (4-D), $B$      & asynchronous, full history & \textbf{0.54} & \textbf{0.83} & \textbf{0.91} & \textbf{0.96} \\
                        & pyABC & 0.45 & 0.74 & 0.85 & 0.90 \\
g-and-k (4-D), $g$      & asynchronous, full history & 0.58 & 0.88 & 0.95 & 0.98 \\
                        & pyABC & \textbf{0.51} & \textbf{0.82} & \textbf{0.91} & \textbf{0.95} \\
g-and-k (4-D), $k$      & asynchronous, full history & 0.56 & \textbf{0.83} & \textbf{0.92} & \textbf{0.96} \\
                        & pyABC & \textbf{0.47} & 0.74 & 0.86 & 0.92 \\
g-and-k (4-D), $A$      & asynchronous, top-$k$ archive & --- & --- & 0.85 & 0.89 \\
\bottomrule
\end{tabular}
\end{table}

\paragraph{Recovery against a reference posterior.} Figure~\ref{fig:reported-recovery} compares the estimated posteriors with available references at successive wall-clock checkpoints. On the Gaussian mean, the asynchronous estimator is closer to the analytic posterior than pyABC's final population at every checkpoint. A rejection sampler given the same simulation budget is as close. With this many draws, simple rejection sampling already performs well in one dimension. On the four-dimensional g-and-k, pyABC is closer to the reference posterior: it ends at $W_1=0.029$ and the asynchronous estimator at $0.069$, with rejection far behind at $0.47$. The effective sample sizes explain the difference: on the Gaussian run, the posterior effective sample size stays at $238$--$296$ at every checkpoint while the history grows from $10^5$ to $1.3\times10^6$ particles. A controlled sweep puts $\mathrm{ESS}/k$ between $1.9$ and $3.6$ for $k\in\{50,100,200\}$. The estimate is reported at the tightest bandwidth reached, which the scheduler sets by requiring that the archive's ESS retain a set fraction of its value (\S\ref{sec:method-tolerance}). Consequently, only a relatively small number of particles contribute substantial weight to the posterior estimate, even after millions of evaluations. The effective support remains on the order of a few multiples of $k$. In one dimension that support is small and stable; in four it is small and erratic ($7$ to $2987$ across checkpoints and replicates), and accuracy follows it. The low and variable effective sample size suggests intermittent weight degeneracy; these results do not establish a systematic bias. The limitation appears once the target is identified well enough for the bandwidth to tighten.

\begin{figure}[t]
\centering
\includegraphics[width=\linewidth]{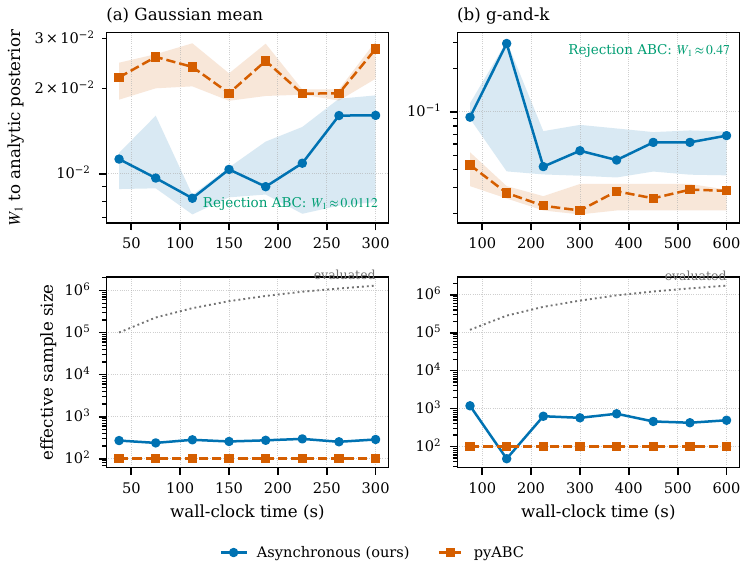}
\caption{Each method's own posterior (\S\ref{sec:experiments-metrics}) scored against the analytic posterior (a) and the g-and-k reference posterior (b) over wall-clock prefixes. The upper panels show median $W_1$ and interquartile ranges over five replicates. The lower panels show posterior effective sample size; dotted lines indicate the number of evaluated particles.}\label{fig:reported-recovery}
\end{figure}

\paragraph{A known truth on the real simulator.} On the two-parameter Cellular Potts setup (one synthetic dataset with a known truth), all three methods recover the truth (Fig.~\ref{fig:cpm-posterior}, Table~\ref{tab:cpm-production}). The division rate is well identified by every method ($91$--$92\%$ contraction against the prior). For cell volume, contraction is $81\%\pm1\%$ for the asynchronous estimator, $72\%\pm4\%$ for pyABC's last generation with its population weights, and $67\%$ for the best $100$ rejection draws. The truth lies inside the central $90\%$ interval of every run (five per adaptive method, one rejection run). The asynchronous estimator is the tightest and has the largest effective sample ($303$ against pyABC's $66$ of $100$, medians).

\begin{figure}[t]
\centering
\includegraphics[width=0.62\linewidth]{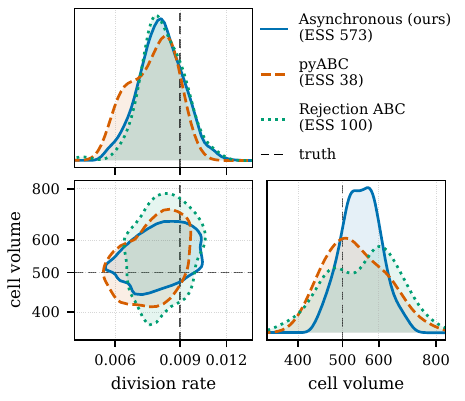}
\caption{Cellular Potts posterior (division rate, cell volume) on the fixed production run, replicate $0$: the asynchronous AMIS estimator over the full history, pyABC's last completed generation with its population weights, and the $100$ best rejection draws, each drawn as a weighted kernel density estimate of the method's own particles (effective sample sizes in the legend) (Table~\ref{tab:cpm-production}). Axes are the log-uniform prior coordinates ticked in physical units; the joint panel outlines each method's $90\%$ highest-density region, and dashed lines mark the truth.}\label{fig:cpm-posterior}
\end{figure}

\begin{table}[t]
\caption{Cellular Potts $50^3$, five replicates per adaptive method and one rejection run, $3600$\,s on $48$ workers (means ${\pm}$ standard deviation for contraction, medians otherwise). The first row is the asynchronous sampler at the benchmark's original starting bandwidth ($10$, twenty times the prior-predictive median discrepancy), the rest at the production value ($0.1$). Rejection ABC is one run of $13{,}000$ prior draws on all $48$ workers, keeping the best $100$. Contraction is $1-\mathrm{sd}/\mathrm{sd}_{\mathrm{prior}}$; the truth lies inside the central $90\%$ interval of every run. Bold marks the best value per column.}\label{tab:cpm-production}
\centering
\footnotesize
\setlength{\tabcolsep}{4pt}
\begin{tabular}{@{}lcccccc@{}}
\toprule
 & & & & \multicolumn{2}{c}{contraction} & \\
\cmidrule(lr){5-6}
 & sims & util. & $\epsilon_{(100)}$, end & division rate & cell volume & ESS \\
\midrule
Asynchronous, $\epsilon_0=10$   & \textbf{13{,}300} & \textbf{99.8\%} & {\boldmath$6.3\times10^{-5}$} & $84\pm6\%$ & $13\pm14\%$ & \textbf{645} \\
Asynchronous, $\epsilon_0=0.1$  & 12{,}900 & \textbf{99.8\%} & $7.6\times10^{-5}$ & {\boldmath$92\pm1\%$} & {\boldmath$81\pm1\%$} & 303 \\
pyABC, $\epsilon_0=0.1$ & 5{,}400 & 42.8\% & $3.1\times10^{-4}$ & $91\pm1\%$ & $72\pm4\%$ & 66 \\
Rejection ABC, best $100$ of $13{,}000$ & 13{,}000 & --- & $2.9\times10^{-3}$ & $91\%$ & $67\%$ & 100 \\
\bottomrule
\end{tabular}
\end{table}

\paragraph{The starting bandwidth.} With expensive simulations, the bandwidth may remain close to its initial value throughout the available budget, leaving the reported posterior sensitive to $\epsilon_0$ even after the sampled parameters have concentrated near the observations. The first row of Table~\ref{tab:cpm-production} is the same sampler started at the benchmark's original $\epsilon_0=10$, twenty times the prior-predictive median discrepancy of $0.46$. It reports $13\%\pm14\%$ contraction on the cell volume, erratic as well as poor, although its draws are as concentrated as the production run's (the end tolerances agree within $20\%$). The runs differ in the bandwidth used for reporting. After the run's $62$ calls per worker, the schedule had not moved from $\epsilon_0$, which ended $321\times$ above the tolerance the draws supported, against $2\times$ at $\epsilon_0=0.1$. The bandwidth transient is negligible when millions of evaluations are available, explaining why it was not observed in the inexpensive benchmarks. It remains a limitation in the regime the method is designed for. The rule of \S\ref{sec:experiments-metrics}, about a fifth of the prior-predictive median, recovers the posterior. Re-reporting a stored history at the $k$-th order statistic of its own discrepancies costs no simulations and reaches a comparable posterior ($93\%/85\%$ contraction on replicate $0$). The synchronized variant shows the same transient on the straggler benchmark, whose Gaussian-mean simulator is near-instantaneous (Appendix~\ref{app:twin-straggler}).

\paragraph{The archive size.} The effective sample size stays near a few multiples of $k$. Increasing $k$ does not necessarily solve this limitation: When we recompute the posterior from the same stored history using order-statistic bandwidths as $k$ increases from $10$ to $1000$, the bandwidth increases and cell-volume contraction falls from $85\%$ to $78\%$. Running with $k=300$ leaves the schedule at $\epsilon_0$ for the whole hour. On the benchmarks with inexpensive simulators, $k=100$ is Pareto-optimal for throughput and calibration (Appendix~\ref{app:moved}). The archive size is a benchmark-dependent choice, made by checking calibration and the bandwidth reached.

\section{Discussion}\label{sec:discussion}
On the tissue simulator, the asynchronous sampler completes more evaluations per hour and reaches a tighter tolerance than the matched pyABC baseline, while also achieving greater contraction of the weakly identified parameter and a larger posterior effective sample size (Table~\ref{tab:cpm-production}). On the one-dimensional Gaussian target, our sampler is calibrated where pyABC under-covers. On the four-dimensional g-and-k benchmark, however, pyABC achieves a smaller Wasserstein distance to the reference posterior. Asynchronous execution is also less effective when evaluations are very low-cost. Together, these measurements suggest four practical guidelines for the workloads and implementation settings examined here.

\emph{Consider asynchronous sampling when evaluations take at least a few milliseconds.} Below that the per-arrival coordination costs more than a barrier does, and a generation-based sampler with a population sized to the machine is the better tool. Above it, removing the barrier recovers a cost that grows with the worker count and with the spread of simulation times (\S\ref{sec:results-boundary}).

\emph{Estimate the synchronization cost from runtime measurements.} The barrier's cost is computable from the asynchronous run's own per-evaluation timing, within the domain stated in \S\ref{sec:experiments-metrics} (\S\ref{sec:results-barrier}).

\emph{Set the starting bandwidth from the prior-predictive discrepancy scale.} For expensive simulations, the initial bandwidth may remain too large throughout the run. We therefore recommend setting it from prior-predictive discrepancies and checking the reported posterior by recomputing it from the completed evaluations at their $k$-th smallest discrepancy (\S\ref{sec:results-posterior}). A floor at $10^{-4}\epsilon_0$ satisfies the positive-bandwidth condition in \S\ref{sec:theory} and would not bind on the reported g-and-k and Cellular Potts runs.

\emph{Choose the archive size with calibration in mind.} Enlarging the archive loosens the final bandwidth and lengthens the transient, and the calibration checks are inexpensive to repeat (\S\ref{sec:results-posterior}).

Increasing the number of completed simulations does not automatically increase the effective sample size of the posterior estimate. In our experiments, that size stays within a few multiples of the archive size even as the history grows. The reporting rule contributes to this limit: for simplicity, the posterior is evaluated at the tightest bandwidth reached by the scheduler, so only a restricted part of the history receives appreciable weight. Adjusting the posterior reporting bandwidth independently of the sampling schedule may improve accuracy without modifying asynchronous execution.

\section{Limitations}\label{sec:limitations}
The consistency results have two important qualifications. First, the estimator converges to the intended smooth-ABC posterior when the fidelity ratio approaches one, under the stated assumptions. Of the four factors contributing to this ratio, two are measured and two bounded. Second, the asymptotic statements apply only when the reporting bandwidth remains bounded away from zero. A floor of $10^{-4}\epsilon_0$ guarantees this. It would have clipped the Gaussian-mean and Lotka--Volterra runs but none of the g-and-k and Cellular Potts runs. Even where the limits apply, they describe the run's continuation, since no reported run reached its floor (\S\ref{sec:theory}). At a fixed deadline, fast simulations are more likely to have finished than slow ones. The completed evaluations may therefore over-represent parameter regions with short runtimes, an instance of length bias in deadline-limited sampling \citep{murray2021anytime}. Under the execution model, the direct effect of evaluations still running at the deadline is limited to at most $W$ candidates, an $O(W/n)$ perturbation (Proposition~\ref{prop:censoring}); the remaining effect operates through the adaptive proposal sequence and is covered by the fidelity ratio $r$. A runtime-coupled study found no measurable shift of the posterior mean at the couplings tested, and we do not correct for the $O(W/n)$ term (Appendix~\ref{app:moved}).

\section{Conclusion}\label{sec:conclusion}
Generation-based ABC can leave workers idle while slow simulations finish. We quantify this synchronization penalty using simulation times recorded during an asynchronous run and validate the prediction against a synchronized variant on synthetic runtime distributions and a tissue simulator. The generation-free sampler removes the barrier and, wherever a simulation costs more than a few milliseconds, converts the simulations it gains into a tighter tolerance. The full-history posterior estimator can be reconstructed from the simulation log. Our analysis characterizes its potential deviation from the smooth-ABC posterior through a fidelity ratio that is only partly measured. The estimator is well calibrated in the one-dimensional tests but conservative for some parameters of the four-dimensional benchmark. What it delivers is bounded by the starting bandwidth and the archive size, and \S\ref{sec:discussion} says how to set them.

\subsubsection*{Broader Impact Statement}
This work is methodology for likelihood-free inference on parallel machines. Its direct effect is on compute: it removes the idle time a generation barrier imposes on heterogeneous simulator workloads, and the prediction of \S\ref{sec:results-barrier} lets a practitioner estimate that saving from their own timing data before spending an allocation. The method does not introduce an application-specific risk beyond those associated with the models and inferences to which it is applied.

\subsubsection*{Reproducibility}
The code is available at \url{https://github.com/ju-he/async-abc}. It contains the implementation (the asynchronous sampler as a Propulate propagator, the matched-kernel pyABC baseline, the benchmarks and runners), every run configuration, the cluster launchers and the test suite, together with the summary data behind every figure and table, the scripts that regenerate each, and the committed output of every diagnostic quoted in the text, so every figure, table and quoted number can be regenerated from the repository alone. The raw per-particle records (${\approx}30$\,GB for the Gaussian-mean experiment alone) and the per-worker timing logs will be released publicly upon publication.

\subsubsection*{Declaration of AI Assistance}
We used AI-based tools (Claude Opus and GPT Astra) to improve the clarity, grammar, and structure of the manuscript. Additionally, Claude Fable and GPT Astra were used to assist in checking mathematical derivations and proofs for correctness and consistency. All AI-assisted suggestions were reviewed by the authors, who take full responsibility for the scientific content, mathematical arguments, and conclusions presented in this work.

\subsubsection*{Acknowledgments}
The authors gratefully acknowledge the Gauss Centre for Supercomputing e.V. (\url{https://www.gauss-centre.eu}) for funding this project by providing compute time through the John von Neumann Institute for Computing on the GCS Supercomputer JUWELS at Forschungszentrum J\"ulich under the grant TissueTwin. A.S. and J.H. recognize support by the Helmholtz Information and Data Science School for Health. M.G., A.S. and J.H. recognize support by the Helmholtz Association's Initiative and Networking Fund under the Helmholtz AI platform grant and the grant ZebraTwin. A.S. recognizes support by the Helmholtz Foundation Model Initiative of the Helmholtz Association under the grants PROFOUND and VirtualCell. The funders had no role in study design, data collection, analysis, decision to publish or preparation of the manuscript.

\bibliography{references}
\bibliographystyle{tmlr}

\appendix

\section{The fidelity ratio}\label{app:theory-more}
This appendix relates the analysis to existing AMIS theory, bounds how far the tilt by $r$ can move the target, and decomposes $r$ and measures part of it on stored runs.

\paragraph{Relation to existing AMIS theory.} For a fixed proposal, the $r$-tilted limit of Corollary~\ref{cor:tilt} is the elementary observation that self-normalized importance sampling against a mismatched denominator converges to a correspondingly tilted target. The present scheme, however, is adaptive at every draw, so the argument runs through a martingale array rather than an i.i.d.\ one. The available AMIS results do not cover it: the convergence argument of \citet{cornuet2012amis} is heuristic, and \citet{marin2019consistency} restrict the adaptation so that the weighting stage decouples. The prior floor in the \emph{denominator} bounds every weight by $1/\delta$ with no moment or density-ratio condition. This lets the method use a local archive proposal at all: a proposal built around the $k$ best particles puts far less mass in the tails than the prior, so the usual bound on the prior-to-proposal density ratio does not hold.

\begin{proposition}[how far the tilt can move the target]\label{prop:tilt}
Let $\zeta_0:=\pi_{\epsilon_\infty}(|r-1|)$. Then
\begin{equation}
\bigl\|\pi^r_{\epsilon_\infty}-\pi_{\epsilon_\infty}\bigr\|_{\mathrm{TV}}
=\frac{\pi_{\epsilon_\infty}\bigl(|r-\pi_{\epsilon_\infty}(r)|\bigr)}{2\,\pi_{\epsilon_\infty}(r)}
\qquad\text{and}\qquad
\zeta_0<1\ \Longrightarrow\ \bigl\|\pi^r_{\epsilon_\infty}-\pi_{\epsilon_\infty}\bigr\|_{\mathrm{TV}}\le\frac{\zeta_0}{1-\zeta_0}.
\label{eq:tilt-bound}
\end{equation}
If $|r-1|\le\eta_A$ on a set $A\subseteq\Theta$ and $|r-1|\le M$ off it, then $\zeta_0\le\zeta:=\eta_A+M\,\pi_{\epsilon_\infty}(A^c)$, and the bound holds with $\zeta$ in place of $\zeta_0$ whenever $\zeta<1$.
\end{proposition}
The identity says that the tilt moves the target through the variation of $r$ over posterior mass, not through its distance from one: $\pi^{cr}_{\epsilon_\infty}=\pi^r_{\epsilon_\infty}$ for every constant $c>0$. The bound replaces the centered absolute moment by the uncentered $\zeta_0$ and the divisor $\pi_{\epsilon_\infty}(r)$ by $1-\zeta_0$, and is loose by exactly that much. Estimating $\zeta_0$ from a stored run requires \emph{untilting}: the normalized weights $\widehat w_i$ target $\pi^r$, not $\pi_{\epsilon_\infty}$, so $\sum_i\widehat w_i|r(\theta_i)-1|$ estimates $\pi^r(|r-1|)$; weights proportional to $W_{i,n}/r(\theta_i)$ target $\pi_{\epsilon_\infty}$ and give $\widehat\zeta_0$.

\subsection{What $r$ is for this implementation}\label{sec:theory-r}
We decompose $r$ by inserting intermediate mixtures between the density actually sampled from and the function used in the denominator, changing one implementation feature at a time. Let $\bar q^{(1)}=\tfrac1n\sum_{i\le n}q_i$ use the \emph{nominal} archive proposals: the truncated mixtures the propagator intends, with exact normalizing constants and without the rejection cap or the underflow redraws, two implementation details described under (d) below. We obtain $\bar q^{(2)}$ by replacing those constants with the implementation's product-of-marginals values, and $\bar q^{(3)}$ by taking the $m$-point draw-proportional version of $\bar q^{(2)}$ carrying the observed share $\nu_n$ of bootstrap (prior) draws. Assume each of the four displayed ratios has a uniform limit, finite and bounded away from zero. Fidelity alone does not give this, since \eqref{eq:fidelity} constrains only their product. Writing $\bar q$ for the $m$-snapshot denominator \eqref{eq:snapshot-denominator},
\[
r=\underbrace{\lim\frac{\bar q^\star}{\bar q^{(1)}}}_{\text{(d) cap, redraws, asynchrony}}\cdot
\underbrace{\lim\frac{\bar q^{(1)}}{\bar q^{(2)}}}_{\text{(c) normalizer}}\cdot
\underbrace{\lim\frac{\bar q^{(2)}}{\bar q^{(3)}}}_{\text{(b) quadrature}}\cdot
\underbrace{\lim\frac{\bar q^{(3)}}{\bar q}}_{\text{(a) floor}} .
\]
The product telescopes for any positive intermediates, so the content is in the choice of intermediates: a different insertion order attributes the total differently, and the factors are a decomposition of one quantity rather than four independent errors.

\emph{(a) The prior floor.} When genuine bootstrap draws carry less than a $\delta$ share of the history, $\lambda_n=\delta>\nu_n$ and the denominator holds more prior mass than the sample does. With $u=q_\infty/\pi$ the limiting proposal-to-prior ratio, this contributes
$r_{\mathrm{floor}}=\bigl(\nu_\infty+(1-\nu_\infty)u\bigr)/\bigl(\delta+(1-\delta)u\bigr)$, so $r_{\mathrm{floor}}-1=(\delta-\nu_\infty)(u-1)/(\delta+(1-\delta)u)$: $r_{\mathrm{floor}}$ is exactly $1$ when $\nu_\infty\ge\delta$, below $1$ wherever the proposal is thinner than the prior, and $r_{\mathrm{floor}}\to1$ as $\delta\to0$, i.e.\ as $m$ grows. This is the one contribution available in closed form (Lemma~\ref{lem:floor}).

\emph{(b) The snapshot quadrature.} The $m$-point mixture stands in for the $n$-point one, with error at most $V_n/m$ where $V_n=\sum_{i<n}\|q_{i+1}-q_i\|_\infty$ is the total variation of the \emph{reconstructed} proposal path (the summed size of its steps); this controls factor (b), not the fidelity of the whole. Fixing $m$ therefore trades cost against fidelity; see the remark below.

\emph{(c) The truncation constant.} Each mixture component is normalized by the product of its per-axis in-box masses, which is exact for a diagonal perturbation covariance and increasingly accurate as a component's mass sits away from the faces of the box relative to its own scale, but an approximation otherwise.

\emph{(d) Capped rejection, underflow redraws, and (asynchronously) the arrival-order reconstruction.} A parent's perturbation is rejection-sampled into the box for at most $992$ attempts, so component $j$, with mixture weight $W_j$ and in-box mass $Z_j$, is realized with probability $W_j[1-(1-Z_j)^{992}]$ and the residual is emitted as a prior draw. A candidate whose proposal-time denominator falls below $10^{-12}$ is redrawn up to five times, the last one being emitted regardless. The cap changes the distribution relative to unlimited rejection whenever $Z_j<1$, through the residual probability $(1-Z_j)^{992}$; $Z_j\approx10^{-3}$ is the scale below which that residual becomes appreciable, not a threshold above which it vanishes. Under asynchronous execution, this factor additionally includes the difference between the proposal each candidate came from, formed on the prefix its worker held, and the snapshots the post-hoc pass rebuilds from an arrival-ordered log. By Lemma~\ref{lem:adapted} the emitted density is the conditional distribution of the candidate, so no further gap arises at that step, and each rebuilt snapshot differs from the proposal in force at that index by at most the $W$ candidates in flight at the time.

\paragraph{Measured, in part.} Two of the four factors, the prior floor and the snapshot quadrature, can be estimated rather than bounded, because $\widehat\pi_n$ is a post-hoc pass: recomputing \eqref{eq:snapshot-denominator} at a much larger $m$ and at the observed prior share gives a reference against which those factors are evaluated at every particle. The truncation constant and factor (d) are bounded rather than measured. The five Cellular Potts production replicates (Appendix~\ref{app:cpm}) have about $12{,}900$ asynchronous evaluations each across $48$ workers and two correlated dimensions, so that factor (c), exact only for a diagonal covariance, differs from one. Both measured factors, (a) and (b), differ from one: the bootstrap draws make up $1.1\%$ of each history against a floor of $\delta=2.4\%$, so (a) contributes $\widehat\zeta=0.011$ and (b) $0.024$--$0.028$, together $0.035$--$0.039$ with $r\in[0.24,1.41]$. The posterior estimate lies $0.2$--$0.5\%$ in total variation from the reference at $m=400$ and the observed prior share. Posterior means agree within $0.012$ standard deviations and marginal widths within $1\%$, narrower in nine of the ten marginals, which is the tail-suppressing direction that (a) predicts; the effective sample size is unchanged. On a one-dimensional Gaussian-mean history, where the exact $n$-point mixture is computable, contributions (a) and (b) together give $\widehat\zeta=0.042$; the total-variation distance from the reference that removes them is $0.26\%$, and there too the estimate is $0.5\%$ narrower. Of contribution (d), the two events a stored history records, a rejection fallback and an exhausted redraw, occurred neither in the $12{,}000$ draws of that history nor in the $63{,}699$ archive-phase draws of the five Cellular Potts replicates. The bound in \eqref{eq:tilt-bound} is loose against both direct measurements by an order of magnitude ($3.6$--$4.0\%$ against $0.2$--$0.5\%$ on Cellular Potts, $4.4\%$ against $0.26\%$ on the one-dimensional history), as sup-norm-free bounds valid for every $f$ generally are. The identity in \eqref{eq:tilt-bound} says what the direct measurements estimate: the centered absolute moment of the removed factors under the reference, divided by twice their mean.

\paragraph{Does the proposal path settle?} Factor (b) is controlled by the total variation $V_n$ of the reconstructed proposal path, which a stored history makes measurable. Rebuilding $q_\tau$ along the five Cellular Potts production replicates, the per-step drift $\|q_{\tau+1}-q_\tau\|_\infty$ decays as $\tau^{-0.7}$ to $\tau^{-0.95}$. The archive's membership changes only $825$--$874$ times in about $12{,}900$ draws, tracking $k\log n$ at a ratio between $0.54$ and $0.92$ across the whole run, the rate a record-statistics heuristic predicts, since the $k$-th best of $n$ draws improves with probability ${\approx}k/n$. Membership changes alone would thus add $O(\log n)$ to the path's total variation if each moves the proposal by a bounded amount; the measured drift, which also includes the bandwidth updates, sums to between $O(n^{0.05})$ and $O(n^{0.3})$. Either way $V_n$ grows sublinearly, the regime in which a snapshot count growing faster than $V_n$ removes factor (b) (Remark~\ref{rem:growing-m}). The measurement spans one benchmark over $1.3$ decades of $n$ per replicate, so it is a diagnostic rather than an asymptotic statement.

\paragraph{Fixed versus growing $m$.} With $m$ fixed, the $m$-point mixture tracks the $n$-point one only if the proposal path settles, the error being $V_n/m$. Letting $m=m_n\to\infty$ with $V_n/m_n\to0$ removes \emph{this} contribution to $r$ (it does not give $r\equiv1$, since the other three factors remain), and the proofs go through for any $m_n=o(n/\log n)$ \emph{provided the floor $\delta$ is held fixed rather than tied to $m$}. The implementation ties them through $\delta=0.5/(m{+}1)$, so growing $m$ there would send $\delta\to0$ and lose the uniform weight bound $1/\delta$ of Lemma~\ref{lem:bounded}, which fixes the constants in the law of large numbers. The post-hoc pass can decouple the two (Appendix~\ref{app:theory}, Remark~\ref{rem:growing-m}). The retroactive pass then costs $\mathcal{O}(n\,m_n\,k)$ instead of $\mathcal{O}(nk)$, and the literal cumulative mixture is the endpoint $m_n=n$ at $\mathcal{O}(n^2k)$. On the Gaussian-mean history the posterior estimate is within $3\times10^{-4}$ in total variation of the exact $n$-point mixture by $m=50$, so a $2.5\times$ more expensive post-hoc pass removes factor (b) almost entirely. The paper reports the $m=21$ end; \S\ref{sec:method-amis} gives its cost.

\section{Proofs}\label{app:theory}

\subsection{Setting}\label{app:theory-setting}
We work with a filtration $(\mathcal{F}_i)$ for which, writing $\mathbb{E}_{i-1}[\cdot]=\mathbb{E}[\cdot\mid\mathcal{F}_{i-1}]$, the pair $\xi_i=(\theta_i,\rho_i)$ is $\mathcal{F}_i$-measurable and
\begin{equation}
\theta_i\mid\mathcal{F}_{i-1}\sim \tilde q_i,
\qquad
\rho_i\mid\mathcal{F}_{i-1}\vee\sigma(\theta_i)\sim p(\cdot\mid\theta_i).
\label{eq:app-conditional-law}
\end{equation}
Particles are indexed by \emph{proposal} time, not arrival time. Under sequential execution $\mathcal{F}_i=\sigma(\xi_1,\dots,\xi_i)$ serves and \eqref{eq:app-conditional-law} holds by construction. Under asynchronous execution a discrepancy $\rho_i$ may still be in flight when candidate $i+1$ is proposed, so the $\sigma$-field of everything realized by that moment does not contain $\xi_i$, and a larger filtration is needed. A two-worker example shows the difficulty. Worker A proposes candidate 1, whose simulation is slow. Worker B proposes candidate 2, receives its result, and proposes candidate 3 while candidate 1 is still running. When candidate 3 is proposed, $\xi_1$ has not been realized. Conditioning on $\xi_1,\xi_2$ instead does not recover the density worker B emitted either: B built it from $\xi_2$ alone, and which results had reached B depends on the durations of the two simulations, which $\xi_1,\xi_2$ do not record. The following lemma supplies a filtration from the execution model that adds the durations, the worker identities and the delivery randomness, and under which the emitted density is the conditional distribution of the candidate.

\begin{lemma}[adapted filtration under the execution model]\label{lem:adapted}
Let Assumption~\ref{ass:filtration} hold, let $w_i$ be the worker that proposes candidate $i$ (in this lemma and its proof $w_i$ is a worker index, not the parent weight of the history) and $D_i$ the duration of its simulation, and let $\mathcal{F}_i:=\sigma\bigl((\theta_j,\rho_j,D_j,w_j)_{j\le i},\,\mathcal{E}\bigr)$, where $\mathcal{E}$ is the $\sigma$-field of the exogenous scheduling randomness of (B3). Then $\xi_i$ is $\mathcal{F}_i$-measurable and \eqref{eq:app-conditional-law} holds with $\tilde q_i$ the proposal density that worker $w_i$ emits on the prefix of the history it has received. Moreover, for every $\tau$ the arrival-ordered prefix of length $\tau$ and the proposal-ordered prefix $\{1,\dots,\tau\}$ differ in at most $W$ candidates each.
\end{lemma}
\begin{proof}
Candidate $i$ is proposed when a worker becomes free, so its proposal time $T_i$ is the completion time $T_j+D_j$ of an earlier candidate $j<i$ of the same worker (or a start time, for the first $W$ candidates). By (B3) and induction on $i$, the time $T_i$, the worker $w_i$ and the set $A_i\subset\{1,\dots,i-1\}$ of candidates whose results have reached $w_i$ by $T_i$ are functions of $(T_j,D_j,w_j)_{j<i}$ and of $\mathcal{E}$. They are therefore $\mathcal{F}_{i-1}$-measurable. The emitted density $q^{\mathrm{em}}_i$ is a function of the history restricted to $A_i$, so it is $\mathcal{F}_{i-1}$-measurable too. By (B2), $\theta_i$ is drawn from $q^{\mathrm{em}}_i$ with randomness independent of $\mathcal{F}_{i-1}$, so the conditional distribution of $\theta_i$ given $\mathcal{F}_{i-1}$ is $q^{\mathrm{em}}_i$: the first part of \eqref{eq:app-conditional-law}, with $\tilde q_i=q^{\mathrm{em}}_i$. By (B1) and (B3), $(\rho_i,D_i)$ given $\theta_i$ is independent of $(\theta_j,\rho_j,D_j,w_j)_{j<i}$ and of $\mathcal{E}$, which generate $\mathcal{F}_{i-1}$; hence $\rho_i\mid\mathcal{F}_{i-1}\vee\sigma(\theta_i)\sim p(\cdot\mid\theta_i)$. For the last claim, let $P$ be the set of candidates proposed before the $\tau$-th arrival and $A$ the set of the first $\tau$ arrivals. Then $A\subseteq P$, $\{1,\dots,\tau\}\subseteq P$ since $|P|\ge\tau$, and $P\setminus A$ is the set in flight at that moment, of size at most $W$; so $\{1,\dots,\tau\}\setminus A\subseteq P\setminus A$ has at most $W$ elements, and $A\setminus\{1,\dots,\tau\}$ has the same number.
\end{proof}

Completion times steer which worker proposes next and which prefix it sees. The lemma says that this steering changes \emph{which} proposal density is used, a predictable choice, and not the distribution of the candidate given that density, which is fixed by fresh randomness. The post-hoc estimator consumes the \emph{arrival}-ordered log of one rank (one MPI process; here every worker is a rank) and reconstructs the $q_{\tau_s}$ from its prefixes. By the last clause of the lemma, each reconstructed snapshot differs from the proposal in force at that index by a boundary of at most $W$ particles. This boundary is the asynchrony contribution to $r$ (\S\ref{sec:theory-r}(d)). The lemma fails if simulator outputs are coupled across calls (a shared seed, say) or if the sampler's stream is correlated with runtimes; both are checkable implementation properties, and Table~\ref{tab:assumption-status} records them for ours.

No structural form is imposed on $\tilde q_i$: it is whatever the archive rule, the capped rejection, the underflow redraws and the bootstrap phase jointly produce, and it enters the proofs only through $\bar q^\star_n=\tfrac1n\sum_{i\le n}\tilde q_i$ of \eqref{eq:draw-mixture}, for which the denominator $\bar q_n$ of \eqref{eq:snapshot-denominator} is an $m$-component stand-in. Let $\mathcal{Q}$ be the set of positive functions of the form \eqref{eq:snapshot-denominator} whose components satisfy Assumption~\ref{ass:class}; nothing below requires an element of $\mathcal{Q}$ to integrate to one. Elements of $\mathcal{Q}$ are parameterized by $p_q=(\text{means},\text{component weights},\text{covariances},\alpha,\lambda)$ ranging over a compact set $P_q\subset\mathbb{R}^{D}$, $D=m\,(kd+k+d(d+1)/2)+m+1$, the last coordinate being the prior share $\lambda\in[\delta,1]$ of \eqref{eq:snapshot-denominator}, on which $\bar q_n$ depends affinely.

\begin{lemma}[bounded weights]\label{lem:bounded}
Under Assumptions~\ref{ass:class}--\ref{ass:defensive}, $0\le W_{i,n}\le1/\delta$ for all $i\le n$ and all $n$, surely, whatever the sampler does.
\end{lemma}
\begin{proof}
All components of \eqref{eq:snapshot-denominator} are non-negative, so $\bar q_n\ge\delta\pi$ pointwise on $\Theta$. Hence $W_{i,n}=\pi(\theta_i)K_{\epsilon_n}(\rho_i)/\bar q_n(\theta_i)\le K_{\epsilon_n}(\rho_i)/\delta\le1/\delta$, the last step by $K_\epsilon\le1$ (Assumption~\ref{ass:kernel}). The prior cancels exactly, so the bound needs neither a density-ratio condition nor a moment condition.
\end{proof}

\begin{lemma}[uniform strong law of large numbers for bounded martingale differences]\label{lem:uslln}
Let $P\subset\mathbb{R}^{D}$ be compact and $\{\psi_p\}_{p\in P}$ measurable functions on a measurable space $\mathcal{X}$ with $\sup_{p}\|\psi_p\|_\infty\le B$ and $|\psi_p(x)-\psi_{p'}(x)|\le\Lambda\|p-p'\|$ for every $x\in\mathcal{X}$. Let $(\xi_i)$ be $\mathcal{X}$-valued with $\xi_i$ $\mathcal{F}_i$-measurable. Then
\[
\sup_{p\in P}\Bigl|\tfrac1n\sum_{i\le n}\bigl(\psi_p(\xi_i)-\mathbb{E}_{i-1}\psi_p(\xi_i)\bigr)\Bigr|\;\longrightarrow\;0
\qquad\text{a.s.}
\]
\end{lemma}
\begin{proof}
Fix $p$. The variables $D^p_i:=\psi_p(\xi_i)-\mathbb{E}_{i-1}\psi_p(\xi_i)$ form a martingale difference sequence with $|D^p_i|\le2B$, so by the Azuma--Hoeffding inequality $\mathbb{P}(|\sum_{i\le n}D^p_i|\ge nt)\le2\exp(-nt^2/8B^2)$ for every $t>0$. Fix $\eta>0$ and let $P_\eta$ be a finite $\eta$-net of $P$, of cardinality at most $C_P\eta^{-D}$. A union bound gives $\mathbb{P}(\max_{p\in P_\eta}|\sum_{i\le n}D^p_i|\ge nt)\le2C_P\eta^{-D}\exp(-nt^2/8B^2)$, which is summable in $n$ (the net is fixed, so its cardinality does not grow with $n$); by Borel--Cantelli, $\limsup_n\max_{p\in P_\eta}|\tfrac1n\sum_{i\le n}D^p_i|\le t$ a.s. For arbitrary $p\in P$ pick $p'\in P_\eta$ with $\|p-p'\|\le\eta$; then $|D^p_i-D^{p'}_i|\le2\Lambda\eta$ by the Lipschitz hypothesis (which also bounds the conditional expectations), so $\limsup_n\sup_{p\in P}|\tfrac1n\sum_{i\le n}D^p_i|\le t+2\Lambda\eta$ a.s. Intersecting the almost-sure events over the countable family $t=1/r$, $\eta=1/j$ with $r,j\in\mathbb{N}$ (here $r$ is an integer index, not the fidelity ratio), gives the claim.
\end{proof}

\begin{lemma}[conditional-mean identity]\label{lem:condmean}
Let $\phi:\Theta\to\mathbb{R}$ be bounded measurable, and let the bandwidth $\epsilon\in(0,\epsilon_0]$ and the positive function $q\ge\delta\pi$ each be either deterministic or $\mathcal{F}_{i-1}$-measurable. Then under \eqref{eq:app-conditional-law},
\[
\mathbb{E}_{i-1}\!\left[\frac{\phi(\theta_i)\,K_\epsilon(\rho_i)}{q(\theta_i)}\right]
=\int_\Theta \phi\,L_\epsilon\,\frac{\tilde q_i}{q}\,d\theta .
\]
\end{lemma}
\begin{proof}
Condition on $\mathcal{F}_{i-1}$, given which $\epsilon$ and $q$ are constants, and integrate first over $\rho_i\sim p(\cdot\mid\theta_i)$, which gives the factor $L_\epsilon(\theta_i)=\mathbb{E}[K_\epsilon(\rho)\mid\theta_i]$, then over $\theta_i\sim\tilde q_i$. Fubini applies because the integrand is bounded by $\|\phi\|_\infty/(\delta\pi_-)$. An $\mathcal{F}_\infty$-measurable $q$, such as the run's final denominator $\bar q_n$, is \emph{not} admissible here; the proof below handles it through uniformity instead. The case $\phi=\pi h$ is the one used below.
\end{proof}

\paragraph{The class used below.} Take $\mathcal{X}=\Theta\times[0,\infty)$, $\xi_i=(\theta_i,\rho_i)$, and, for a fixed bounded measurable $h$,
\begin{equation}
\psi_{(\epsilon,p_q)}(\theta,\rho):=\frac{\pi(\theta)\,h(\theta)\,K_\epsilon(\rho)}{q_{p_q}(\theta)},
\qquad
(\epsilon,p_q)\in P_a:=[a,\epsilon_0]\times P_q,\qquad a\in(0,\epsilon_0] .
\label{eq:app-class}
\end{equation}
Each $P_a$ is compact \emph{and deterministic}, which matters: Lemma~\ref{lem:uslln} cannot be applied over a parameter set that depends on the run, and $[\epsilon_\infty,\epsilon_0]$ would be such a set when $\epsilon_\infty$ is random. We therefore apply the lemma on $P_{1/j}$ for every $j\in\mathbb{N}$ with $1/j\le\epsilon_0$ and intersect the countably many almost-sure events; since $\epsilon_n\downarrow\epsilon_\infty>0$ a.s., each path has $\epsilon_n\in[1/j,\epsilon_0]$ for all $n$ once $1/j\le\epsilon_\infty$, so the resulting uniform bound covers $(\epsilon_n,\bar q_n)$ on that path. No deterministic bandwidth floor is required. On $P_a$, $\sup\|\psi\|_\infty\le\|h\|_\infty/\delta$ by the computation of Lemma~\ref{lem:bounded}. For the Lipschitz hypothesis: each truncated component density $\varphi(\theta;\mu_j,\Sigma)/Z(\mu_j,\Sigma)$ has parameter-derivatives bounded uniformly over $\theta\in\Theta$, because $\Theta$ is compact, $\sigma_-^2I\preceq\Sigma\preceq\sigma_+^2I$ bounds $\varphi$ and its derivatives, and the truncation mass is bounded below, $Z(\mu_j,\Sigma)\ge c_0>0$, uniformly for $\mu_j\in\Theta$. The implementation's constant is the \emph{product of per-axis} in-box masses, which is where the box form of $\Theta$ in Assumption~\ref{ass:class} enters. Each factor is $\Phi((\mathrm{hi}_d-\mu_{jd})/\sigma_d)-\Phi((\mathrm{lo}_d-\mu_{jd})/\sigma_d)$. Since $\mu_{jd}$ lies in $[\mathrm{lo}_d,\mathrm{hi}_d]$, that interval contains a half-width $w_d/2$ on one side of $\mu_{jd}$, giving a factor at least $\Phi(w_d/2\sigma_+)-\tfrac12>0$. The product over coordinates, $\prod_{d}\bigl(\Phi(w_d/2\sigma_+)-\tfrac12\bigr)$, is therefore bounded below by a constant depending only on the box and $\sigma_+$ and is smooth in $(\mu_j,\Sigma)$ with bounded derivatives there. (The exact joint mass is likewise bounded below, as a positive continuous function of $(\mu_j,\Sigma)$ on a compact set, so nothing in the argument depends on which normalizer is used; the box is needed to define the implementation's constant, not to bound it.) Hence every $q\in\mathcal{Q}$ (a set of positive \emph{functions}, not necessarily normalized; Assumption~\ref{ass:class}) obeys $\delta\pi_-\le q\le q_+<\infty$. The map $q\mapsto q(\theta)$ is Lipschitz in $p_q$ uniformly in $\theta$ and depends affinely on the prior share $\lambda$, with coefficients bounded by $\pi_+$ and $q_+$. Because $q$ is bounded away from $0$, $1/q$ is also Lipschitz. With $K_\epsilon$ Lipschitz in $\epsilon$ (Assumption~\ref{ass:kernel}) and $\pi h$ bounded, \eqref{eq:app-class} satisfies the hypotheses of Lemma~\ref{lem:uslln}. Note that on $\{1/j\le\epsilon_\infty\}$ we have $(\epsilon_n,\bar q_n)\in P_{1/j}$ for every $n$, and that $\bar q_\infty\ge\delta\pi$ by taking uniform limits.

\subsection{Proof of Theorem~\ref{thm:consistency} and Corollaries~\ref{cor:tilt}--\ref{cor:exact}}\label{app:theory-consistency}
The prior floor bounds every weight (Lemma~\ref{lem:bounded}). For a fixed bandwidth and denominator, the weighted average minus its predictable compensator is an average of bounded martingale differences and vanishes (Lemma~\ref{lem:uslln}). The class of bandwidths and denominators is compact and the weights are Lipschitz in its parameters, so this convergence is uniform over the class. Uniformity allows the run-dependent pair $(\epsilon_n,\bar q_n)$ to be substituted, after which Lemma~\ref{lem:condmean} identifies the compensator as the tilted integral.

Work on the event $\{1/j\le\epsilon_\infty\}$ for a fixed $j$, so that $\widehat p_n:=(\epsilon_n,\bar q_n)\in P_{1/j}$ for every $n$; the conclusions then follow on the union over $j$, which is $\{\epsilon_\infty>0\}$. Fix a bounded measurable $h$ and define, for each \emph{fixed} $p\in P_{1/j}$, $\Gamma_i(p):=\mathbb{E}_{i-1}[\psi_p(\xi_i)]$ and $\bar\Gamma_n(p):=\tfrac1n\sum_{i\le n}\Gamma_i(p)$. On the almost-sure event of Lemma~\ref{lem:uslln} for $P_{1/j}$ the bound holds simultaneously for all $p\in P_{1/j}$, so the \emph{function} $\bar\Gamma_n$ may be evaluated at the random $\widehat p_n$:
\begin{equation}
\frac1n\sum_{i\le n}W_{i,n}h(\theta_i)
=\bar\Gamma_n(\widehat p_n)+o(1)
\overset{\text{Lem.\ \ref{lem:condmean}}}{=}\int_\Theta \pi\,h\,L_{\epsilon_n}\,\frac{\bar q^\star_n}{\bar q_n}\,d\theta+o(1)
\quad\text{a.s.},
\label{eq:app-consistency-main}
\end{equation}
using $\tfrac1n\sum_{i\le n}\tilde q_i=\bar q^\star_n$. This is \eqref{eq:tracking}. The middle term is the function $\bar\Gamma_n(\cdot)$ evaluated at $\widehat p_n$; it must not be read as $\mathbb{E}_{i-1}[\psi_{\widehat p_n}(\xi_i)]$, because $\widehat p_n$ depends on draws after $i$ and is not $\mathcal{F}_{i-1}$-measurable, so Lemma~\ref{lem:uslln} does not control that conditional expectation. For \eqref{eq:tracking-ratio}, apply \eqref{eq:tracking} with $h=f$ and with $h\equiv1$, and write $A_n=a_n+o(1)$ and $B_n=b_n+o(1)$, with $A_n,B_n$ the two weighted averages and $a_n,b_n$ the two integrals. If $\liminf_nB_n>0$ then $b_n\ge c>0$ for all large $n$, and $|A_n/B_n-a_n/b_n|\le|A_n-a_n|/B_n+|a_n/b_n|\,|b_n-B_n|/B_n\to0$, since $|a_n/b_n|\le\|f\|_\infty$. This proves the theorem.

\emph{Corollary~\ref{cor:tilt}.} Split the integral in \eqref{eq:tracking} against $\int\pi hL_{\epsilon_\infty}r$:
\[
\Bigl|\int \pi h L_{\epsilon_n}\tfrac{\bar q^\star_n}{\bar q_n}-\int\pi hL_{\epsilon_\infty}r\Bigr|
\le \|h\|_\infty\,\pi_+|\Theta|\,\Bigl[\,R\,L_K(1/j)\,|\epsilon_n-\epsilon_\infty|
+\Bigl\|\tfrac{\bar q^\star_n}{\bar q_n}-r\Bigr\|_\infty\Bigr],
\]
on the event $\{1/j\le\epsilon_\infty\}$ of the localization above, using $|L_{\epsilon_n}-L_{\epsilon_\infty}|\le L_K(1/j)|\epsilon_n-\epsilon_\infty|$ (Assumption~\ref{ass:kernel} and Jensen; every $\epsilon_n$ and $\epsilon_\infty$ lies in $[1/j,\epsilon_0]$ there) and $L_{\epsilon_\infty}\le1$. Here $R=r_++1$ bounds $\bar q^\star_n/\bar q_n$ for all large $n$ by \eqref{eq:fidelity}, with no structural bound on $\tilde q_i$ needed. Both bracketed terms vanish, the first because $\epsilon_n\to\epsilon_\infty$ and the second by \eqref{eq:fidelity}. Hence $\tfrac1n\sum_iW_{i,n}h(\theta_i)\to\int\pi hL_{\epsilon_\infty}r$ a.s.\ on the event. Applying this with $h=f$ and with $h\equiv1$ and dividing, which is legitimate because $Z_r=\int\pi L_{\epsilon_\infty}r\ge r_-Z_{\epsilon_\infty}>0$, gives $\widehat\pi_n(f)\to\pi^r_{\epsilon_\infty}(f)$. Corollary~\ref{cor:exact} is the case $r\equiv1$. (Theorem~\ref{thm:consistency} does not use the monotonicity of $\epsilon_n$, only $\epsilon_\infty>0$; for a schedule that may loosen one localizes on $\{\epsilon_n\ge1/j\text{ for all large }n\}$ instead.)\hfill$\square$

\begin{remark}\label{rem:weaker-c4}
Theorem~\ref{thm:consistency} uses nothing about the form of $\tilde q_i$, and Corollary~\ref{cor:tilt} adds only \eqref{eq:fidelity}. Both are therefore statements about the algorithm as implemented, not about an idealized variant of it: any sampler whose draw mixture $\bar q^\star_n$ is tracked by the denominator $\bar q_n$ in the ratio sense of \eqref{eq:fidelity} is covered by the corollary, the archive rule, the capped rejection and the underflow redraws included.
\end{remark}

\subsection{Deadline censoring}\label{app:theory-censoring}
At the wall-clock deadline the run has proposed $n'$ candidates, of which the set $R$ still in flight, $|R|\le W$, has no discrepancy; the estimate is computed from the $n=n'-|R|$ completed ones. Let $\widehat\pi_{n'}(f)$ be the estimate the launched set would give at the same bandwidth and denominator, i.e.\ with the sums in \eqref{eq:posterior-estimator} extended over all $n'$ candidates, the unobserved discrepancies included.
\begin{proposition}[deadline censoring]\label{prop:censoring}
Under Assumptions~\ref{ass:kernel}--\ref{ass:defensive}, for every bounded $f$,
\[
\bigl|\widehat\pi_n(f)-\widehat\pi_{n'}(f)\bigr|
\le\frac{2\|f\|_\infty\,|R|/\delta}{\sum_{i\le n'}W_{i,n}}
\le\frac{2\|f\|_\infty\,W}{\delta\,n\,\overline W_n},
\qquad \overline W_n:=\tfrac1n\textstyle\sum_{i\le n}W_{i,n} .
\]
The censoring is therefore negligible whenever $W/n\to0$ and the average weight stays bounded below, and it is $o(n^{-1/2})$ whenever $W=o(\sqrt n)$.
\end{proposition}
\begin{proof}
Write $W_i$ for $W_{i,n}$, $A=\sum_{i\le n',\,i\notin R}W_if(\theta_i)$, $B=\sum_{i\le n',\,i\notin R}W_i$, $a=\sum_{i\in R}W_if(\theta_i)$ and $b=\sum_{i\in R}W_i$. Then $\widehat\pi_{n'}(f)-\widehat\pi_n(f)=(A+a)/(B+b)-A/B=\bigl(a-(A/B)\,b\bigr)/(B+b)$, and $|a-(A/B)b|\le\sum_{i\in R}W_i\,|f(\theta_i)-A/B|\le2\|f\|_\infty\,b$. Lemma~\ref{lem:bounded} gives $b\le|R|/\delta\le W/\delta$, and $B+b\ge B=n\overline W_n$.
\end{proof}
The launched set is the one the filtration of Lemma~\ref{lem:adapted} indexes, so Theorem~\ref{thm:consistency} and its corollaries apply to $\widehat\pi_{n'}$, and the proposition transfers them to the completed-set estimate the implementation reports. The length bias of deadline-limited sampling (\S\ref{sec:limitations}) thus acts on the estimate in two ways. The bound controls the direct effect through the at most $W$ censored candidates; the indirect effect acts through the proposal sequence, whose archive is formed from arrived particles and so leans toward regions that simulate quickly. This second effect is part of $\tilde q_i$ and is corrected by the importance weighting up to $r$. On the five Cellular Potts production replicates, $W=48$ and $n\approx12{,}900$, so $W/n\approx0.4\%$; the worst-case factor $1/\delta\approx42$ makes the displayed bound loose at that budget, and the drain-after-deadline experiment of Appendix~\ref{app:moved}, which lets the in-flight simulations finish, measures the effect on the posterior mean at $0.013$ or less.

\begin{remark}[growing snapshot count]\label{rem:growing-m}
Nothing above requires $m$ to be fixed except the compactness and fixed dimension $D$ of $P_q$, and the fixed floor $\delta$. The implementation ties $\delta$ to $m$ as $0.5/(m{+}1)$, so the two must be decoupled before $m$ is allowed to grow, since Lemma~\ref{lem:bounded}'s bound is $1/\delta$. With $\delta$ fixed and $m=m_n\to\infty$, run the proof of Lemma~\ref{lem:uslln} with an $n$-dependent net: the union bound costs $D_n\log(1/\eta_n)$ with $D_n=\Theta(m_n)$, and the Lipschitz constant over the mixture weights grows at most polynomially in $m_n$, so choosing $\eta_n\to0$ fast enough leaves the Borel--Cantelli step summable for any $m_n=o(n/\log n)$. In that regime factor (b) of $r$ (\S\ref{sec:theory-r}) vanishes whenever $V_n/m_n\to0$, with $V_n=\sum_{i<n}\|q_{i+1}-q_i\|_\infty$ the total variation of the proposal path: the draw-proportional $m_n$-point mixture is then a converging quadrature of the empirical mixture. The retroactive pass costs $\mathcal{O}(n\,m_n\,k)$.
\end{remark}

\subsection{Proof of Proposition~\ref{prop:tilt}}\label{app:theory-tilt}
Write $\pi$ for $\pi_{\epsilon_\infty}$. For bounded $f$, since $\pi^r(f)=\pi(fr)/\pi(r)$ and $\pi(r-\pi(r))=0$,
\[
\pi^r(f)-\pi(f)=\frac{\pi\bigl(f\,(r-\pi(r))\bigr)}{\pi(r)}
=\frac{\pi\bigl((f-c)\,(r-\pi(r))\bigr)}{\pi(r)}\quad\text{for every constant }c .
\]
The supremum of the modulus over $\|f\|_\infty\le1$ is attained at $f=\operatorname{sign}(r-\pi(r))$ and equals $\pi(|r-\pi(r)|)/\pi(r)$; halving gives the identity. For the bound, the first form also reads $[\pi(f(r-1))-\pi(f)\pi(r-1)]/\pi(r)$, whose modulus is at most $2\|f\|_\infty\zeta_0/(1-\zeta_0)$ because $\pi(r)\ge1-\zeta_0$; the supremum over $\|f\|_\infty\le1$, halved, is $\zeta_0/(1-\zeta_0)$. Under the set hypothesis $\zeta_0\le\eta_A\pi(A)+M\pi(A^c)\le\zeta$, and $x\mapsto x/(1-x)$ increases on $[0,1)$. For $\zeta_0\ge1$ only the trivial $\|\cdot\|_{\mathrm{TV}}\le1$ survives.\hfill$\square$

\begin{lemma}[prior-floor factor]\label{lem:floor}
With $u=q_\infty/\pi$ the limiting proposal-to-prior ratio and $\nu_\infty<\delta$ the limiting bootstrap share, the prior floor's contribution to $r$ in \S\ref{sec:theory-r}(a) is
\[
r_{\mathrm{floor}}-1=\frac{\nu_\infty+(1-\nu_\infty)u}{\delta+(1-\delta)u}-1=\frac{(\delta-\nu_\infty)(u-1)}{\delta+(1-\delta)u},
\]
and on $\{u\ge c\}$, splitting at $u=1$ and bounding each branch, $|r_{\mathrm{floor}}-1|\le(\delta-\nu_\infty)\max\{(1-\delta)^{-1},(\delta+(1-\delta)c)^{-1}\}$, which is the $\eta_A$ of Proposition~\ref{prop:tilt} for $A=\{u\ge c\}$. The factor is identically $1$ when $\nu_\infty\ge\delta$, and $\to1$ as $\delta\to0$.
\end{lemma}
\begin{proof}
Evaluate \eqref{eq:fidelity} for the floor alone: the draw mixture is $\nu_\infty\pi+(1-\nu_\infty)q_\infty$ and the denominator $\delta\pi+(1-\delta)q_\infty$; divide both by $\pi$ and subtract one.
\end{proof}
\subsection{Status of each assumption in the implementation}\label{app:theory-status}
\begin{table}[h]
\caption{What the implementation guarantees, what it assumes, and what it violates. ``By construction'' means the code enforces the property for every run, not that it was checked post hoc.}\label{tab:assumption-status}
\centering
\footnotesize
\setlength{\tabcolsep}{3pt}
\begin{tabular}{@{}p{0.16\textwidth}p{0.13\textwidth}p{0.675\textwidth}@{}}
\toprule
Assumption & Status & Basis \\
\midrule
\ref{ass:kernel} kernel, Lipschitz part & by construction & Gaussian and Epanechnikov are peak-normalized and Lipschitz in $\epsilon$ on every $[a,\epsilon_0]$, with $L_K(a)=2/ea$ and $2/a$. The hard indicator is not, and is excluded from the theory. \\
\ref{ass:kernel} $Z_{\epsilon_\infty}>0$ & holds for Gaussian & The Gaussian kernel is strictly positive, so $L_{\epsilon_\infty}>0$ and the limiting target is non-degenerate whatever the discrepancy distribution. For a compactly supported kernel (Epanechnikov) it requires the achievable discrepancies to reach inside the support, which we do not verify. \\
\ref{ass:class} denominator class & by construction & The prior support is a box and the archive lies in it, so component means do; the covariance carries an additive jitter $10^{-9}\mathrm{tr}(\widehat\Sigma)/d+10^{-12}\bar b^2$ ($\bar b^2$ the mean squared box width), whose second term is a run-independent floor, giving $\sigma_-^2>0$, and $\sigma_+^2$ follows from the archive lying in a bounded box. The truncation \emph{normalizer} is the product of per-axis in-box masses, so with a correlated covariance $\bar q_n$ does not integrate to one. Assumption~\ref{ass:class} does not require that: it constrains a positive function, and the discrepancy is carried by $r$ instead (\S\ref{sec:theory-r}(c)). \\
\ref{ass:filtration} execution model & holds; implementation property & (B1), (B2): the propagator draws from its own stream, seeded once per rank, and every simulator call receives a fresh seed derived from the run seed, the rank and a per-rank call counter, so the two streams are independent in the sense in which pseudo-random streams are, and no two calls share a seed. (B3): a worker proposes when its own simulation completes and builds the proposal from the results that have reached it; delivery is asynchronous message passing, exogenous to both streams. Lemma~\ref{lem:adapted} then gives the adapted filtration with $\tilde q_i$ the emitted density. What asynchrony leaves in $r$ is the arrival-order reconstruction of the snapshots, a boundary of at most $W$ particles per snapshot (\S\ref{sec:theory-r}(d)). \\
\ref{ass:defensive} prior floor & by construction & The retroactive denominator floors the prior component's mass at $\delta=0.5/(m{+}1)$ (Appendix~\ref{app:implementation}), so $\bar q_n\ge\delta\pi$ and Lemma~\ref{lem:bounded} bounds every posterior weight by $1/\delta$, independently of what the sampler does. \\
\ref{ass:stab} fidelity \eqref{eq:fidelity} & assumed; two of its four factors measured & The one place the sampler enters for Corollary~\ref{cor:tilt}; Theorem~\ref{thm:consistency} does not use it. Decomposed into four factors in \S\ref{sec:theory-r}, two of them estimated from stored histories rather than bounded; since the factors multiply, the measured part is a partial accounting rather than a bound on the whole. \\
$\{\epsilon_\infty>0\}$ (the event of the limit statements) & \textbf{floor never reached} on g-and-k and Cellular Potts; not met on Gaussian mean and Lotka--Volterra & The event concerns the reporting bandwidth only; the sampler's own bandwidths enter through $\tilde q_i$, on which no condition is placed. A strictly positive \texttt{min\_tol} forces $\epsilon_\infty\ge\texttt{min\_tol}>0$, since $\{\epsilon_n\}$ is non-increasing and clipped from below, which is all Theorem~\ref{thm:consistency} and Corollary~\ref{cor:tilt} need. The floor acts only through that clip, so a run whose final bandwidth lies above it is path-identical to the floored run. Final bandwidths over $\epsilon_0$: g-and-k $2.2$--$3.0\times10^{-3}$, Cellular Potts production $1.3$--$3.2\times10^{-3}$ (and $1.2\times10^{-3}$ or more at $\epsilon_0=10$; the $80^3$ runs never left $\epsilon_0$), g-and-k calibration trials $1.1\times10^{-2}$--$0.55$ (a replayed sample of $30$ of the $1{,}000$, on $6$ workers rather than $48$), all above a floor of $10^{-4}\epsilon_0$; Gaussian mean $6$--$7\times10^{-6}$ and Lotka--Volterra $1.3$--$2.4\times10^{-7}$, below it. The bandwidth a finite budget reaches is not an asymptotic $\epsilon_\infty$: no run reached its floor, and the effective-sample-size ceiling (\S\ref{sec:theory}) is the consequence. \\
\bottomrule
\end{tabular}
\end{table}

\section{Implementation Details}\label{app:implementation}

\paragraph{Relation to classical ABC-SMC.}\label{sec:method-comparison}
Table~\ref{tab:method-comparison} sets the method against the generation-staged ABC-SMC/PMC implementation we compare with (the pyABC baseline), property by property. None of the listed properties is intrinsic to ABC-SMC as a family (smooth kernels and richer proposals exist), but each characterizes pyABC as configured. Because all state is a pure function of the history (last row), there is nothing to synchronize between workers, which is what lets the workers run without waiting for one another (\S\ref{sec:method-overview}).
\begin{table}[h]
\caption{A typical generation-staged ABC-SMC/PMC implementation (the pyABC baseline used here) versus the proposed generation-free method. The listed properties characterize pyABC as configured, not ABC-SMC as a family. The running sampler does carry the online snapshot buffer of \S\ref{sec:method-amis}; the posterior estimate rebuilds it from the history.}\label{tab:method-comparison}
\centering
\begin{tabular}{@{}lll@{}}
\toprule
Property & Classical ABC-SMC & This work \\
\midrule
Update style & generation, barrier-synchronized & event-driven (single-arrival), no barrier \\
Archive & explicit, generational & reconstructed, sliding top-$k$ \\
ABC likelihood & hard threshold & smooth kernel \\
Importance weight & against $q_{t-1}$ only & balance heuristic over $\mathcal{S}$ \\
Archive state & mutable, carried & none; reconstructed per call \\
Posterior estimate & final population & pure function of the history \\
\bottomrule
\end{tabular}
\end{table}

\paragraph{Carried state and restart.} Between calls the running sampler carries only the kept snapshots of \S\ref{sec:method-weights} and the scheduler's two controls: the search cadence and the acceptance gate described below. A restart therefore reproduces the estimator over whatever history was evaluated, not the trajectory the run would have taken (kill-and-resume, Appendix~\ref{app:protocol}).

\paragraph{Hyperparameters.} Table~\ref{tab:hyperparameters} lists every setting of the sampler with its option name in the implementation and the value used in the experiments.
\begin{table}[h]
\caption{Settings of the asynchronous sampler. Option names are those of the implementation; values are those of the production experiments, except where a study sweeps them (Appendix~\ref{app:moved}).}\label{tab:hyperparameters}
\centering
\footnotesize
\begin{tabular}{@{}>{\raggedright\arraybackslash}p{0.07\textwidth}>{\raggedright\arraybackslash}p{0.30\textwidth}>{\raggedright\arraybackslash}p{0.22\textwidth}>{\raggedright\arraybackslash}p{0.32\textwidth}@{}}
\toprule
Symbol & Meaning & Option & Value \\
\midrule
$k$ & archive size & \texttt{k} & $100$ ($200$ in the four-dimensional calibration runs) \\
$\epsilon_0$ & starting bandwidth & \texttt{tol\_init} & $5$ (Gaussian mean), $10$ (g-and-k), $5\times10^{5}$ (Lotka--Volterra), $0.1$ (Cellular Potts; $10$ in its earlier parameterization) \\
$K_\epsilon$ & ABC kernel & \texttt{kernel} & Gaussian \\
$s$ & scale of the perturbation covariance $\Sigma_n$ & \texttt{perturbation\_scale} & $0.8$ \\
$S$ & kept snapshots & \texttt{amis\_snapshots} & $20$ \\
 & snapshot cadence, in calls & \texttt{amis\_interval} & $k$ \\
 & ESS retention per search & \texttt{ess\_target} & $0.95$ \\
 & largest tightening per search & \texttt{max\_tighten\_factor} & $0.5$ \\
 & search cadence, in calls per rank & \texttt{bisect\_interval} & $k$ \\
$\mu$ & acceptance gate: hold until $k+\mu$ particles lie within $\epsilon$ & \texttt{additional\_needed\_inds} & $k$ \\
 & bandwidth floor; a positive value puts the run on $\{\epsilon_\infty>0\}$ (\S\ref{sec:theory}) & \texttt{min\_tol} & $10^{-4}\epsilon_0$ recommended; unset in the reported runs, where that value is never reached on g-and-k or Cellular Potts (\S\ref{sec:theory}) \\
$\delta$ & prior floor in $\bar q_n$ & --- & $0.5/(m{+}1)$ \\
\bottomrule
\end{tabular}
\end{table}

\paragraph{Glyphs of Fig.~\ref{fig:algorithm-strip}.} The toy problem has $k=N=12$. Row (a) is textbook PMC with a quantile schedule and a hard threshold, weights $\omega^{(t+1)}_i\propto\pi(\theta_i)/q_t(\theta_i)$ normalized over the $N$ kept. Point sizes vary where weights do (a1--a3 by $\omega^{(t)}$, a6 by $\omega^{(t+1)}$, b3 by the archive weights $\tilde W_j$); shading in b1 is the discrepancy (light: far); dotted ellipses are $\Sigma_t$ and $\Sigma_n$; gray contours are the current proposal in a5 and b5 and the stored snapshots in b6 (oldest lightest); the dashed ellipse in row (a) is the hard acceptance region; the inset in column 2 is the kernel over $\rho$, the Gaussian $K_{\epsilon_n}$ in (b) against the indicator in (a). The thin box in every panel is the prior box, identical throughout. Bootstrap draws are prior draws with $w^\star=1$ and no stored bandwidth.

\paragraph{Bandwidth schedule.} The schedule is a function of the evaluated history and is \emph{monotone by construction}: the propagator clips any proposed bandwidth above the history-reconstructed running minimum $\epsilon_{\text{hist}}$ (\S\ref{sec:method-tolerance}), so the effective bandwidth can only decrease. Let $\mu$ denote a margin of additional particles (default $\mu=k$); it is unrelated to the snapshot count $m$ of \S\ref{sec:theory}. The bandwidth is held until $k+\mu$ evaluated particles lie within it (the acceptance gate). After that, a search proposes the largest $\epsilon$ at which the kernel-weighted effective sample size of those particles retains a target fraction (default $0.95$) of its current value, tightening by at most \texttt{max\_tighten\_factor} per search. The search runs on a rank's first call, then once per \texttt{bisect\_interval} ($=k$) of its own calls, and again whenever the rank's cached history is rebuilt after an out-of-order arrival; asynchronous arrivals cause such rebuilds and lockstepped ones do not, which is why the synchronized variant's bandwidth tightens less often than the asynchronous sampler's at equal call counts (\S\ref{sec:results-barrier}). The implementation also offers quantile, geometric-decay and acceptance-rate rules (\texttt{scheduler\_type}); they act only under the hard kernel, and with a smooth kernel each reduces to the rule above, so the option has no effect in any reported run.

\paragraph{Proposal covariance.} The archive $A_n$ is the top-$k$ history members by smallest discrepancy. Component weights are formed in log-space as $\tilde W_j \propto w_j\,K_{\epsilon_n}(\rho_j)$ and normalized (a uniform fallback is used if all smooth weights underflow). The perturbation covariance is $\Sigma_n = s\,[\,\widehat{\mathrm{Cov}}_{\tilde W}(A_n) + \lambda_n I\,]$ with $s$ the \texttt{perturbation\_scale} and a two-part jitter $\lambda_n = 10^{-9}\,\mathrm{tr}\,\widehat{\mathrm{Cov}}_{\tilde W}(A_n)/d + 10^{-12}\,\bar b^2$, where $\bar b^2$ is the mean squared prior-box width; it is Cholesky-factored once per call. The first term is scale-aware, so a tightly converged archive (marginal standard deviation ${\sim}10^{-4}$) is not swamped by the floor as a fixed $10^{-6}$ would swamp it; the second is a run-independent absolute floor that keeps the matrix factorable for a degenerate archive and supplies the uniform lower spectral bound Assumption~\ref{ass:class} needs (Table~\ref{tab:assumption-status}). Each mixture component is truncated to the prior box, with its log normalizing mass computed as the product of the per-axis in-box masses. That is exact for a diagonal covariance, in which case the proposal integrates to one on the support; for a correlated covariance it is an approximation, so the evaluated proposal is a positive function that need not integrate to one. This is why \S\ref{sec:theory} states its denominator assumption for positive functions and assigns the discrepancy to factor~(c) of $r$ (\S\ref{sec:theory-r}).

\paragraph{Matched pyABC acceptor.} For pyABC, pyABC's \texttt{UniformAcceptor} is replaced, for smooth kernels, by a probabilistic-rejection acceptor that accepts a candidate of discrepancy $\rho$ with probability $K_{\epsilon}(\rho)$ (peak-normalized so $K_{\epsilon}(0)=1$), using the \emph{same} kernel $K_{\epsilon}$ as the asynchronous sampler; the hard kernel falls back to \texttt{UniformAcceptor}. The rejection step of each replicate draws from an independently seeded random-number generator. This is the sense in which the kernel and acceptor are matched across methods (\S\ref{sec:experiments-baselines}).

\paragraph{Bounded-memory posterior estimator.} The retroactive AMIS estimator \eqref{eq:posterior-estimator} evaluates the snapshot denominator $\bar q_n$ at each of the $n$ evaluated particles. The denominator is a draw-proportional mixture of $m\le S{+}1$ history-reconstructed proposals ($21$ at the $S=20$ of every experiment) and a defensive prior component (mass floored at $\delta=0.5/(m{+}1)$). The total cost is $\mathcal{O}(n\,m\,k)=\mathcal{O}(nk)$ with $m$ fixed. The snapshot indices are spread over the $n_{\text{arch}}$ archive-phase calls: $0,\text{step},2\,\text{step},\dots$, with $\text{step}=\max(1,\lfloor n_{\text{arch}}/S\rfloor)$. The most recent archive-phase proposal is added if it is not already among them. Hence $m\le S{+}1$ rather than $S$, and $m=2$ when $S=0$. Each is rebuilt from the history prefix the call saw, so the pass never reads the online snapshot buffer of \S\ref{sec:method-amis}. For inexpensive simulators the history reaches $\mathcal{O}(10^7)$ particles, so the pass is accumulated in chunks of $2^{16}=65{,}536$ history points: peak memory is $\mathcal{O}(\text{chunk}\cdot m\cdot k)$ rather than $\mathcal{O}(n)$. Each particle's denominator is evaluated independently within a single chunk (there is no log-sum-exp reduction \emph{across} chunk boundaries), so the result is bit-identical to the unchunked computation. Benchmarked at $k=100$, $m=20$ over $n$ from $2.5\times10^4$ to $10^6$, the pass costs $49.3$~s per $10^6$ particles on a single core and is linear in $n$, so the $3\times10^7$-particle Gaussian history extrapolates to ${\approx}25$~minutes; peak memory is flat in $n$, as the chunking predicts. The pass runs after the timed budget and is excluded from the throughput measurements; for an inexpensive simulator it is comparable to the timed budget itself, for expensive simulators negligible against the simulation time.

\section{Additional Experimental Protocol}\label{app:protocol}

\paragraph{Effective simulating ranks.} A ``$W$-worker'' point allocates $W$ physical cores to both samplers. Every rank of the asynchronous sampler evaluates; pyABC's MPI sampler is root-driven, so rank $0$ dispatches without simulating and pyABC has $W{-}1$ simulating ranks. We report the allocated $W$ for both and do not correct the resulting advantage of $1/W$ for the asynchronous sampler ($2.1\%$ at $W{=}48$, $0.3\%$ at $384$).

\paragraph{Seeding.} Each experiment derives its replicate seeds deterministically from a single base seed via \texttt{make\_seeds}: a \texttt{random.Random(base)} stream emits unique non-negative 31-bit integers, one per replicate, portable across platforms and independent of any global random-number-generator state. Derived seeds (observed-data generation, the pyABC rejection step) use a structured BLAKE2b hash (\texttt{stable\_seed}) of their inputs, so identical inputs reproduce identical seeds across fresh and resumed (\texttt{--extend}) runs.

\paragraph{Crash recovery (kill-and-resume).} On the Gaussian-mean benchmark ($4$ workers, a $1.2\times10^{5}$-simulation budget) a run hard-killed partway through inference and relaunched with an identical command reloads the periodic history checkpoint, completes the same total simulation budget with no duplicated records, and recovers the analytic posterior mean (absolute error $7.6\times10^{-3}$, against $7.5\times10^{-4}$ for an uninterrupted run). It is not a bit-identical replay (the two posterior means differ by $6.9\times10^{-3}$), because the sampler's random-number generator is re-seeded rather than checkpointed, MPI arrival order is not deterministic, and the online AMIS snapshot buffer is rebuilt empty on restart. The post-restart proposal sequence therefore follows a different, equally valid adaptation trajectory: recovery yields a valid posterior over the actual evaluated history, not the identical one (\S\ref{sec:method-constraints}).

\paragraph{Cluster and wall-clock control.} Experiments run on JUWELS (Forschungszentrum J\"ulich) under ParaStation MPI, 48 cores per node, with jobs queued through SLURM. All wall-clock-budgeted methods stop at a fixed wall-clock budget enforced uniformly by a monotonic-clock deadline. Once any rank trips the deadline, it serializes its current state into the records and returns, so a partial final generation is recorded rather than discarded. The simulation-limited studies marked $\dagger$ in Table~\ref{tab:configs}, the simulation-based calibration (SBC) and kill-and-resume runs, instead stop at a fixed per-trial simulation count. Posterior estimation and plotting run off the timed inference path, so they do not count against the budget. Matched-budget tolerances (\S\ref{sec:experiments-metrics}) count only records of completed simulations; pyABC also writes one record per accepted particle, and mixing the two would flatter it.

\paragraph{Benchmark configurations.} Table~\ref{tab:configs} lists the resolved settings ($k$ is the archive/population size; the budget is per method run). The strong-scaling studies sweep the worker count from $1$ to $288$ (Lotka--Volterra; sub-node $1/4/16$ and fully-packed multiples of $48$ up to six nodes) and to $384$ (Cellular Potts; up to eight nodes), with archive size $k\in\{100,1000\}$; all other studies fix the workers shown.

\begin{table}[t]
\caption{Resolved experiment configurations. Workers, per-run wall-clock budget, replicates (trials for SBC), archive size $k$, and methods (A: asynchronous; S: pyABC (matched kernel); R: rejection ABC; T: synchronized variant). The synchronized variant row covers the straggler, heterogeneity and Cellular Potts $50^3$ studies at $16$, $48$ and $48$--$384$ workers. $^\dagger$~Simulation-limited: the budget is per-trial (SBC) or per-run (kill-and-resume) \emph{simulations}, not wall-clock seconds; the synchronized variant, and the asynchronous re-run matched to it, stop at a fixed simulation count per level (\S\ref{sec:experiments-baselines}, Table~\ref{tab:twin}). $^\ddagger$~Also run with AMIS on/off, a drain-after-deadline variant, and a steeper coupling regime (base delay $0.5$\,s, cap $8$\,s). \textsuperscript{\S}~The Cellular Potts rows of these studies ran the benchmark's earlier configuration (Appendix~\ref{app:cpm}).}\label{tab:configs}
\centering
\begin{tabular}{@{}lccccc@{}}
\toprule
Experiment & Workers & Budget (s) & Reps & $k$ & Methods \\
\midrule
Gaussian mean          & 48 & 300  & 5   & 100 & A,S,R \\
g-and-k                & 48 & 600  & 5   & 100 & A,S,R \\
Lotka--Volterra        & 48 & 600  & 5   & 100 & A,S \\
Cellular Potts $50^3$  & 48 & 3600 & 5   & 100 & A,S,R \\
Cellular Potts $80^3$  & 48 & 3600 & 3   & 100 & A,S,R \\
Straggler              & 16 & 300  & 5   & 100 & A,S \\
Runtime heterogeneity  & 48 & 60   & 5   & 100 & A,S \\
Parameter-coupled runtime$^\ddagger$ & 48 & 60 & 5 & 100 & A,S \\
Strong scaling (LV)    & 1--288 & 180  & 5 & 100/1000 & A,S \\
Strong scaling (CPM)\textsuperscript{\S} & 1--384 & 1800 & 3--5 & 100 & A,S \\
Synchronized variant\textsuperscript{\S}  & 16 / 48 / 48--384 & sims$^\dagger$ & 5 & 100 & T,A \\
Sensitivity            & 48 & 300  & 5   & 50/200 & A \\
Ablation               & 48 & 300  & 5   & 100 & A \\
SBC (Gaussian, 1-D)    & 48 & 20k$^\dagger$  & 1000 & 100 & A,S \\
SBC (bimodal)          & 48 & 20k$^\dagger$  & 1000 & 100 & A,S \\
SBC (g-and-k, 4-D)     & 48 & 40k$^\dagger$  & 1000 & 200 & A,S \\
SBC (g-and-k, $2\times$) & 48 & 80k$^\dagger$ & 500 & 200 & A,S \\
Kill-and-resume        & 4  & 120k$^\dagger$ & 1    & 100 & A \\
\bottomrule
\end{tabular}
\end{table}

\paragraph{Benchmark-specific notes.} Each benchmark fixes one observed dataset (generated from seed $42$) and infers parameters under uniform priors; the discrepancy is a distance between simulated and observed summary statistics. \emph{Gaussian mean:} $n_{\text{obs}}=100$ observations $\sim\mathcal{N}(0,1)$, prior $\mu\sim\mathcal{U}[-5,5]$, summary the sample mean, discrepancy $|\bar x_{\text{sim}}-\bar x_{\text{obs}}|$; the closed-form Gaussian posterior is the calibration and recovery reference. \emph{g-and-k:} the quantile-defined g-and-k distribution ($c=0.8$), $n_{\text{obs}}=1000$, true $(A,B,g,k)=(3,1,2,0.5)$, priors $A\sim\mathcal{U}[0,6]$, $B\sim\mathcal{U}[0.1,4]$, $g\sim\mathcal{U}[0,5]$, $k\sim\mathcal{U}[0,1]$; summary the seven octiles (the $0.125,\dots,0.875$ quantiles), discrepancy Euclidean. \emph{Lotka--Volterra:} a Gillespie stochastic simulation \citep{gillespie1977} of the four-reaction predator--prey system from $(X_0,Y_0)=(50,100)$ over $T_{\max}=30$, true rates $\theta=(0.5,0.025,0.025,0.5)$, priors $\theta_1,\theta_4\sim\mathcal{U}[0.05,2]$ and $\theta_2,\theta_3\sim\mathcal{U}[0.001,0.2]$; summary the six statistics $(\overline X,\operatorname{sd}X,\overline Y,\operatorname{sd}Y,\log(1{+}\operatorname{var}X/\overline X),\log(1{+}\operatorname{var}Y/\overline Y))$, discrepancy Euclidean. \emph{Cellular Potts:} the real \texttt{cellsInSilico} (NAStJA) simulator; the setup and the screening that fixed it are in Appendix~\ref{app:cpm}. Its per-evaluation cost dominates, which is why it uses the longest wall-clock budget.

\paragraph{Computing environment.} All experiments ran on the JUWELS Cluster at the J\"ulich Supercomputing Centre: dual-socket Intel Xeon Platinum 8168 (Skylake) nodes with $48$ physical cores and ${\approx}94$\,GB memory per node, connected by Mellanox EDR InfiniBand. The software stack is ParaStationMPI~$5.10$ (the 2025 GCC toolchain), Python~$3.12.3$, NumPy~$1.26.4$, SciPy~$1.13.1$, mpi4py~$4.0.1$, and pyABC~$0.12.17$, with POT \citep{flamary2021pot} for the sliced Wasserstein distance; the asynchronous sampler runs as a propagator in a fork of Propulate. Unless noted, ranks are placed one per physical core, so a ``$W$-worker'' point occupies $W$ physical cores packed onto $\lceil W/48\rceil$ fully-reserved nodes. The production experiment of the three benchmarks with inexpensive simulators (Gaussian mean, g-and-k, Lotka--Volterra) instead placed its $48$ ranks as $12$ per node on four nodes under whole-node reservation because the single-node submissions failed. This placement is why its asynchronous throughput is below the packed single-node scaling run's on the same benchmark (\S\ref{sec:results-boundary}). The Cellular Potts production run and the straggler, heterogeneity, synchronized variant and scaling experiments are single-node or fully packed. Each strong-scaling point reports five replicates (three to five for the largest Cellular Potts points), and figures show medians with inter-quartile ranges.

\section{The Cellular Potts inference setup}\label{app:cpm}
The benchmark runs the \texttt{cellsInSilico} Cellular Potts tissue model \citep{berghoff2020cis} through the NAStJA framework \citep{berghoff2018nastja}: a $50^3$ lattice for $501$ steps (an $80^3$ lattice for $1001$ steps in the larger configuration), seeded with a small cluster of cells that divide and move, summarized at the final step. The inference problem was designed from a screening experiment of about $12{,}000$ simulations at parameter settings chosen to cover the space of six candidate parameters. We scored each summary block, a scalar summary or a ten-dimensional curve of summary statistics, by the ratio of between-parameter to within-parameter variation. The configuration we started from weighted seven summary blocks equally. Two scalars, the log cell count and the log $95$th-percentile radius, carry nearly all of the signal: their signal-to-noise ratios are $12$ and $9$, against about $1$ for the four ten-dimensional curve blocks. We therefore compute the discrepancy from those two scalars, weighted equally. One evaluation is the feature-average of four replicate simulations at the same parameters, because a single realization is too noisy for the discrepancy to order nearby parameters. The two summaries resolve exactly two directions: how many cells there are and how large the cluster is at a given count. We therefore infer the two parameters that move those directions independently, the division rate (count) and the target cell volume (radius at fixed count); their responses are near-orthogonal ($|\cos|=0.07$). Every other candidate parameter moves the summaries only through the cell count. At $50^3$, cell motility, adhesion (how energetically favorable contact between two cells is) and surface tension (how energetically favorable a spherical cell shape is) are confounded with the division rate at $|\cos|\ge0.8$. At $80^3$, with $6.4\times$ the cells, the confounding grows stronger rather than breaking ($0.81\to0.98$ for the surface term). The benchmark is therefore two-dimensional as a property of the model, not for want of screening. Priors are log-uniform on $[0.001,0.2]$ and $[200,1200]$ with truths $0.009$ and $500$; motility is fixed at $1400$. Rejection ABC on $2000$ prior draws forecast $91\%$ and $64\%$ contraction for the two parameters before any adaptive run, against the $92\%$ and $81\%$ the asynchronous estimator reports (\S\ref{sec:results-posterior}).

\paragraph{Two configurations.} The \emph{production} configuration is the two-parameter setup above and supplies every posterior result, Tables~\ref{tab:matched-eps} and \ref{tab:cpm-production}, the $80^3$ systems comparison, and the theory diagnostics of Appendix~\ref{app:theory-more}. The \emph{earlier} configuration was used by the $50^3$ synchronized variant and scaling experiments (the $50^3$ points of Fig.~\ref{fig:barrier}c, Table~\ref{tab:twin-cpm}, Fig.~\ref{fig:scaling} right, Fig.~\ref{fig:cpm-util}). It inferred the cell division rate and motility on normalized $[0,1]$ priors mapping to $[6{\times}10^{-5},0.6]$ and $[0,10^4]$ (true normalized values $0.0499$ and $0.2$). Each evaluation comprised one simulation ($5.1$\,s at $50^3$, in-run coefficient of variation $0.10$) and used a whitened distance over morphological summary features: radial density and pair-correlation profiles, free-area and structure-factor profiles, a DBSCAN (density-based spatial clustering of applications with noise) dispersion fraction, and log cell count and radius. The screening above found those two parameters confounded, which is why the production configuration replaced them. The systems measurements depend on the configuration only through its runtime distribution, which the barrier-cost predictor takes from the run itself.

\section{Additional results}\label{app:moved}

\subsection{Runtime heterogeneity: archive quality}

\paragraph{The synthetic instance.} The heterogeneity study shows the same conversion of throughput into accuracy where the runtime distribution is known by construction (Fig.~\ref{fig:hetero-quality}): at a fixed $60$\,s budget on the Gaussian benchmark pyABC completes fewer simulations as the runtime spread $\sigma$ grows and its archive's posterior-mean error to the analytic posterior climbs from ${\approx}0.05$ to ${\approx}0.28$, while the asynchronous sampler completes several-fold more and its archive holds ${\approx}0.007$--$0.011$. These are archive means; the posterior estimates of both samplers in this experiment are scored in Table~\ref{tab:twin-hetero}.

\begin{figure}[h]
\centering
\includegraphics[width=\linewidth]{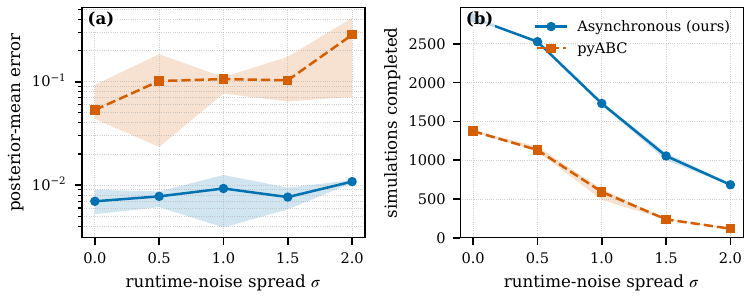}
\caption{Inference per unit wall clock under runtime heterogeneity (Gaussian mean, analytic posterior, $48$ workers, $60$\,s budget, five replicates; bands are inter-quartile ranges). Left: error of the archive mean (the unweighted top-$k$ archive, not $\widehat\pi_n$; \S\ref{sec:results-posterior}) to the analytic posterior mean, against the runtime spread $\sigma$. Right: simulations completed within the budget. The barrier starves pyABC as $\sigma$ grows; the asynchronous sampler's error stays flat.}\label{fig:hetero-quality}
\end{figure}

\subsection{The synchronized variant on the straggler benchmark, and pyABC under injected heterogeneity}\label{app:twin-straggler}
\paragraph{A persistent straggler.} One of $16$ workers is given a permanent post-evaluation delay of $0.1$\,s scaled by $0,1,5,10,20$ (Table~\ref{tab:twin}). The generation time is the slow worker's recorded evaluation time ($2.000$\,s at $20\times$) plus the $3.6$\,ms per-evaluation overhead measured on the fast workers. The synchronized variant therefore cannot exceed $16/2.004=7.98$ simulations per second. Dividing the measured asynchronous rate by that ceiling gives $401\times$, against a measured $402\times$. The agreement is the same at $5\times$ and $10\times$. At $0\times$ and $1\times$ the prediction fails: with no straggler to wait for, the collective itself costs about $0.3$\,s per barrier on a $4$\,ms workload, a term the model does not contain, and taking the barrier only once per $112$ evaluations (pyABC's population) spreads that cost over more evaluations. From $5\times$ up the two barrier granularities agree with each other and with the prediction within $0.5\%$ ($0.2\%$ at $10\times$, $0.06\%$ at $20\times$): once a straggler sets the generation time, the barrier's cost does not depend on how often it is taken.

\begin{table}[h]
\caption{Synchronized variant on the straggler benchmark ($W=16$, five replicates, medians). Throughput in simulations per second; \emph{predicted} as in \S\ref{sec:experiments-metrics}, and at $0\times$ and $1\times$ the collective's latency, absent from the model, dominates. Lower block: $W_1$ of each sampler's posterior to the analytic posterior at the synchronized variant's simulation count and seeds, at the bandwidth its schedule reached and, last row, re-reported at $\epsilon_{(100)}$ (\S\ref{sec:results-posterior}).}\label{tab:twin}
\centering
\begin{tabular}{@{}lccccc@{}}
\toprule
Straggler slowdown & $0\times$ & $1\times$ & $5\times$ & $10\times$ & $20\times$ \\
\midrule
Asynchronous, sims/s              & 3888 & 3112 & 3241 & 3218 & 3205 \\
Synchronized variant, barrier every $W$           & 48.7 & 65.3 & 31.6 & 15.9 & 8.0 \\
Synchronized variant, barrier every $112$         & 162.8 & 148.4 & 31.8 & 15.9 & 8.0 \\
Synchronized variant, predicted                   & (3870) & (154) & 31.8 & 15.9 & 7.98 \\
\midrule
Ratio, measured (every $W$)       & $80\times$ & $48\times$ & $103\times$ & $202\times$ & $402\times$ \\
Ratio, predicted                  & --- & --- & $102\times$ & $202\times$ & $401\times$ \\
\midrule
\multicolumn{6}{@{}l}{\emph{Posterior recovery, $W_1$ to the analytic posterior, at the synchronized variant's simulation count}} \\
Evaluations ($=$ calls per rank $\times16$) & 75{,}104 & 44{,}512 & 12{,}304 & 6{,}304 & 3{,}200 \\
Bandwidth reached, asynchronous   & $5\times10^{-4}$ & $7\times10^{-4}$ & $0.003$ & $0.006$ & $0.16$ \\
Bandwidth reached, synchronized variant every $W$ & $5\times10^{-4}$ & $8\times10^{-4}$ & $0.06$ & $1.3$ & $2.9$ \\
$\epsilon_{(100)}$ of the same history & $3\times10^{-4}$ & $5\times10^{-4}$ & $0.002$ & $0.003$ & $0.008$ \\
Asynchronous at that count        & 0.011 & 0.009 & 0.010 & 0.010 & 0.068 \\
Synchronized variant, barrier every $W$           & 0.013 & 0.007 & 0.016 & 0.789 & 1.713 \\
Synchronized variant, barrier every $112$         & 0.009 & 0.009 & 0.009 & 0.017 & 1.082 \\
Synchronized variant every $W$, re-reported at $\epsilon_{(100)}$ & 0.015 & 0.010 & 0.010 & 0.010 & 0.010 \\
\bottomrule
\end{tabular}
\end{table}

\paragraph{The schedule must move.} The synchronized variant shows the bandwidth transient of \S\ref{sec:results-posterior} on a near-instantaneous simulator. At the synchronized variant's own simulation count and seeds the asynchronous sampler reports $W_1=0.010$ and $0.068$ to the analytic posterior at $10\times$ and $20\times$ where the synchronized variant reports $0.79$ and $1.71$ (Table~\ref{tab:twin}, lower block). The propagator is the same, and so are the $200$ calls per rank at $20\times$. But the synchronized variant's schedule reached a bandwidth of $2.9$, against $0.16$, on histories whose $100$th-smallest discrepancy is $0.008$, because under lockstep it takes fewer search steps (Appendix~\ref{app:implementation}). Re-reporting the synchronized variant's histories at $\epsilon_{(100)}$ brings them to $0.010$--$0.015$.

\begin{figure}[h]
\centering
\includegraphics[width=\linewidth]{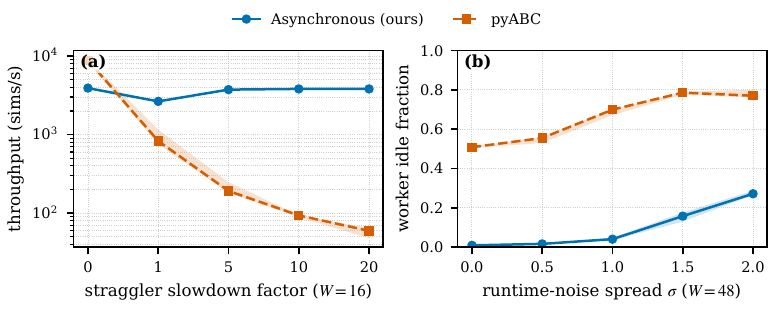}
\caption{pyABC against the asynchronous sampler under injected runtime heterogeneity on the Gaussian mean (\S\ref{sec:results-barrier}). (a) Throughput under a persistent straggler: one of $16$ workers carries a permanent post-evaluation delay of $0.1$\,s scaled by the factor ($300$\,s; medians and inter-quartile ranges over five replicates; logarithmic axis). (b) Fraction of worker time spent idle under a lognormal runtime multiplier of spread $\sigma$ ($48$ workers, $60$\,s).}\label{fig:barrier-pyabc}
\end{figure}

\subsection{The synchronized variant under runtime heterogeneity and on Cellular Potts}
Table~\ref{tab:twin-hetero} gives the synchronized variant's heterogeneity experiment behind \S\ref{sec:results-barrier}, with the posteriors of every sampler at the bandwidth its schedule reached and re-reported at the order statistic (\S\ref{sec:results-posterior}); Table~\ref{tab:twin-cpm} factorizes the synchronized variant's Cellular Potts slowdown into the barrier's share and the unattributed duration ratio.
\begin{table}[h]
\caption{Synchronized variant under runtime heterogeneity ($W=48$, lognormal runtime multiplier of scale $\sigma$, five replicates, medians of simulations per second of each sampler's own elapsed wall clock). \emph{Predicted} is from the injected runtime distribution and the asynchronous run's measured overhead. The $\sigma=0$ control costs almost nothing because these evaluations take $1$\,s and barrier latency is negligible against them; the cost is the wait for the slowest draw and grows with $\sigma$. pyABC degrades less than the synchronized variant because it discards latecomers instead of waiting for them. The lower block scores the posteriors against the analytic posterior at the bandwidth each schedule reached, then re-reported at the $100$th order statistic of the same history (\S\ref{sec:results-posterior}); no rank made $k$ calls in the asynchronous run, so its schedule never ran at $\sigma\le0.5$, and the synchronized variant's ran only at $\sigma\le0.5$.}\label{tab:twin-hetero}
\centering
\small
\setlength{\tabcolsep}{4pt}
\begin{tabular}{@{}lccccc@{}}
\toprule
Heterogeneity $\sigma$ & $0$ & $0.5$ & $1.0$ & $1.5$ & $2.0$ \\
\midrule
Asynchronous, sims/s              & 47.6 & 42.1 & 28.9 & 17.6 & 11.4 \\
Synchronized variant, barrier every $W$           & 44.8 & 15.1 & 4.37 & 1.06 & 0.37 \\
Synchronized variant, barrier every $96$          & 45.1 & 19.8 & 6.50 & 1.75 & 0.45 \\
pyABC                             & 23.1 & 18.9 & 9.91 & 4.35 & 2.39 \\
\midrule
Ratio, measured (every $W$)       & $1.06\times$ & $2.8\times$ & $6.6\times$ & $16.6\times$ & $29.6\times$ \\
Ratio, predicted                  & $1.00\times$ & $2.77\times$ & $6.34\times$ & $14.0\times$ & $35.4\times$ \\
Ratio, pyABC                      & $2.1\times$ & $2.2\times$ & $2.9\times$ & $4.1\times$ & $4.8\times$ \\
\midrule
\multicolumn{6}{@{}l}{\emph{Posterior recovery, $W_1$ to the analytic posterior (calls per rank: async $59$--$13$, synchronized variant $440$--$10$)}} \\
Evaluations, async                & 2832 & 2527 & 1731 & 1055 & 684 \\
Evaluations, synchronized variant                 & 21{,}120 & 6{,}912 & 2{,}112 & 672 & 480 \\
Bandwidth reached, async          & 5 & 5 & 0.92 & 1.25 & 0.78 \\
Bandwidth reached, synchronized variant every $W$ & 0.36 & 2.4 & 5 & 5 & 5 \\
Async, at the schedule's bandwidth & 2.12 & 2.04 & 0.68 & 0.93 & 0.56 \\
Synchronized variant every $W$, at the schedule's bandwidth & 0.19 & 1.43 & 2.20 & 2.22 & 2.28 \\
pyABC, last generation            & 0.53 & 0.61 & 1.36 & 2.19 & 2.21 \\
Async, re-reported at $\epsilon_{(100)}$ & 0.011 & 0.011 & 0.010 & 0.013 & 0.042 \\
Synchronized variant every $W$, re-reported at $\epsilon_{(100)}$ & 0.015 & 0.013 & 0.012 & 0.040 & 0.203 \\
Async archive (top $100$)         & 0.018 & 0.017 & 0.016 & 0.014 & 0.019 \\
\bottomrule
\end{tabular}
\end{table}
\begin{table}[h]
\caption{The synchronized variant's slowdown on Cellular Potts $50^3$ (earlier parameterization, Appendix~\ref{app:cpm}), factorized (five replicates, medians). Throughput is utilization over mean simulation time, so the ratio is the product of the two columns that follow. The utilization ratio is what the barrier costs; it agrees with the expected-maximum-to-mean runtime ratio measured within the synchronized variant's generations and with the prediction from the asynchronous run's runtime distribution. The duration ratio (the synchronized variant's simulations running longer, with replicate means from $6$ to $36$\,s against $5.1$\,s on every asynchronous replicate) is not a barrier effect and is not attributed here.}\label{tab:twin-cpm}
\centering
\begin{tabular}{@{}lcccc@{}}
\toprule
Workers ($=$ cores) & 48 & 96 & 192 & 384 \\
\midrule
Throughput ratio, async / synchronized variant        & $1.91\times$ & $1.87\times$ & $1.86\times$ & $2.82\times$ \\
\quad $=$ utilization ratio           & $1.18\times$ & $1.19\times$ & $1.21\times$ & $1.21\times$ \\
\quad $\times$ simulation-duration ratio & $1.63\times$ & $1.59\times$ & $1.52\times$ & $2.33\times$ \\
Straggler factor of the synchronized variant's generations & 1.21 & 1.28 & 1.21 & 1.29 \\
Predicted from asynchronous timing    & $1.26\times$ & $1.30\times$ & $1.32\times$ & ($2.34\times$) \\
\bottomrule
\end{tabular}
\end{table}

\subsection{Tolerance curves}
Figure~\ref{fig:eps-curves} shows the curves Table~\ref{tab:matched-eps} is read from.
\begin{figure}[h]
\centering
\includegraphics[width=0.85\linewidth]{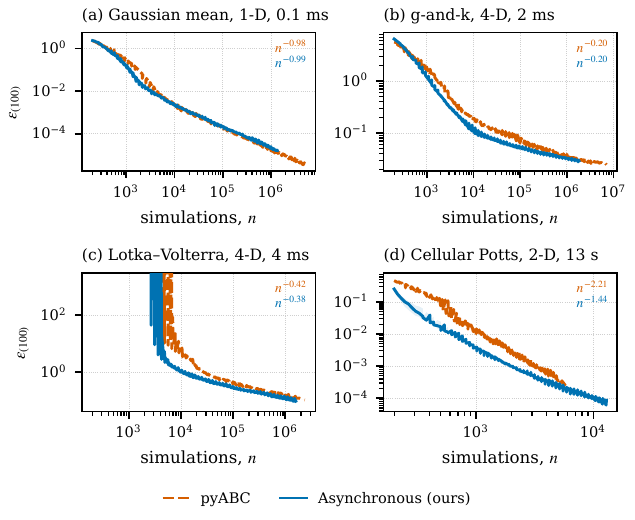}
\caption{Tolerance against simulations spent, $\epsilon_{(100)}(n)$, asynchronous against pyABC on every benchmark (medians and replicate range over five replicates; the slope over the last decade of $n$ in the legend). The vertical gap at equal $n$ is the per-simulation efficiency, the horizontal extent is the extra simulations the wall clock allowed. On Lotka--Volterra $98\%$ of simulations return the extinction fallback, so $\epsilon_{(100)}$ sits on that plateau until about $5000$ draws have produced $100$ survivors; the plateau is cut off.}\label{fig:eps-curves}
\end{figure}

\subsection{Parameter-coupled runtime}
The heterogeneity study above slows simulations independently of their parameters; a sharper test couples runtime \emph{to the inferred parameter}, probing the length bias of deadline-limited sampling (\S\ref{sec:limitations}). On the Gaussian-mean benchmark ($48$ workers, $60$\,s budget, five replicates), the post-evaluation delay grows with the inferred mean $\mu$ as $d(\mu)=\min\!\big(d_0\,e^{\,c(\mu-\mu_c)/\ell},\,d_{\max}\big)$ with $d_0=0.02$\,s, $\mu_c=0$, $\ell=0.15$ and $d_{\max}=0.5$\,s. We sweep $c\in\{0,0.5,1,2\}$; $c=0$ gives uniform runtime. One region of parameter space therefore evaluates far faster than its mirror image, and keeping every evaluated particle over-represents that region in the raw archive. The coupling has the intended effect on the sampling rate: as $c$ increases, aggregate throughput falls by ${\approx}43\%$ for the asynchronous sampler, and the decrease with $c$ is larger for pyABC (Fig.~\ref{fig:param-bias}a).

The posterior does not follow. The asynchronous posterior-mean error to the analytic posterior stays ${\approx}0.01$ at every coupling, comparable to pyABC, which discards latecomers (Fig.~\ref{fig:param-bias}b). This posterior-mean result holds with or without AMIS ($S{=}0$), for both the unweighted archive and the AMIS-reweighted estimator \eqref{eq:posterior-estimator}. It also holds under a stronger runtime coupling (base delay $0.5$\,s, cap $8$\,s). The reweighting changes the spread, not the mean: it lowers the effective-sample-size fraction from ${\approx}0.99$ to ${\approx}0.88$ and widens the central $90\%$ interval by ${\approx}15\%$. A drain-after-deadline variant, which lets every in-flight simulation finish past the deadline, moves the weighted posterior mean by at most $0.013$, within the replicate spread and with no systematic direction, so the deadline censoring is negligible at these couplings, as Proposition~\ref{prop:censoring} predicts for $W/n\ll1$. AMIS neither creates nor removes a mean bias here, and a correction in the manner of anytime Monte Carlo \citep{murray2021anytime} is future work.

\begin{figure}[t]
\centering
\includegraphics[width=0.92\linewidth]{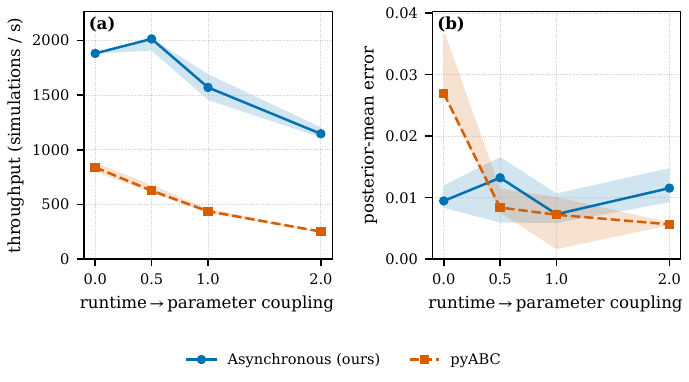}
\caption{Parameter-coupled runtime on the Gaussian-mean benchmark (analytic posterior; $48$ workers, $60$\,s budget, five replicates; bands are inter-quartile ranges). The post-evaluation delay grows with the inferred mean, so one region of parameter space is sampled far more often than its mirror image. (a) Aggregate throughput falls as the coupling steepens (monotonically for pyABC, and from a peak at $c=0.5$ for the asynchronous sampler), consistent with the fast region dominating the raw sample counts. (b) Posterior-mean error to the analytic posterior nonetheless does not increase with coupling and stays comparable between the asynchronous sampler and pyABC, which discards latecomers, so the over-representation does not measurably distort the posterior estimate at these coupling levels.}\label{fig:param-bias}
\end{figure}

\subsection{Strong scaling on an inexpensive simulator}
Lotka--Volterra is an adversarial simulator for a scaling study. Each evaluation takes ${\approx}4$\,ms, so there is little per-evaluation compute to distribute, and an instrumented run (Fig.~\ref{fig:lv-timing}) shows a worker inside the simulator ${\approx}25\%$ of its time at the $48$-worker throughput peak, falling to ${\approx}8\%$ and ${\approx}2\%$ as the job spreads from one node to three to six. The method keeps one history shared by every worker, so every completed evaluation is exchanged with all workers and each can propose from the current archive. This all-to-all exchange has a per-arrival cost that grows with the worker count. Aggregate throughput therefore rises almost linearly while the job fits on one node and then declines across fully packed multi-node jobs (Table~\ref{tab:scaling}), because adding ranks beyond one node lengthens the exchange faster than it adds useful work. pyABC is given a population equal to the worker count at high core counts (a smaller population cannot occupy every core); it scales steadily and overtakes the asynchronous sampler beyond three nodes. This is the boundary of the asynchronous advantage on an inexpensive, homogeneous simulator, and it lifts once the simulator carries real per-evaluation cost (Fig.~\ref{fig:scaling}, right). The archive size moves the boundary, because the per-arrival proposal reconstruction is $\mathcal{O}(k)$. Raising $k$ to $1000$ costs $39\%$ of the single-node peak throughput. Sweeping $k$ over $\{50,100,200,400,800,1000\}$ shows that throughput falls smoothly and monotonically as $k$ grows, and that the node count at which pyABC overtakes moves by a single step, from beyond three nodes to three nodes at $k\ge400$ (Table~\ref{tab:k-frontier}).

\begin{table}[t]
\caption{Strong scaling on Lotka--Volterra: median simulation throughput (sims/s) at a 180\,s budget, over five replicates (Fig.~\ref{fig:scaling}, left, shows the inter-quartile spread). Asynchronous (ours) at both archive sizes vs.\ pyABC. Multi-node points use \emph{fully packed} nodes (worker counts that are exact multiples of $48$); at counts above $100$ pyABC is given \emph{population size equal to the worker count} so it can occupy every core (a population of $100$ leaves the surplus idle), which is why it keeps scaling here. The asynchronous sampler peaks on one node and then declines as multi-node coordination dominates, so pyABC overtakes it beyond three nodes at $k=100$. Enlarging the archive tenfold costs throughput at every scale, because the per-arrival proposal reconstruction is $\mathcal{O}(k)$, and moves that crossover in from beyond three nodes to three.}\label{tab:scaling}
\centering
\begin{tabular}{@{}lcccccccc@{}}
\toprule
Workers & 1 & 4 & 16 & 48 & 144 & 192 & 240 & 288 \\
Nodes & $<$1 & $<$1 & $<$1 & 1 & 3 & 4 & 5 & 6 \\
\midrule
Synchronous & 105 & 96 & 437 & 2952 & 3757 & 4322 & 4539 & 6034 \\
Asynchronous ($k=100$) & 149 & 369 & 1983 & 6240 & 4195 & 2615 & 2088 & 1729 \\
Asynchronous ($k=1000$) & 106 & 133 & 620 & 3804 & 2804 & 1957 & 1587 & 1336 \\
\bottomrule
\end{tabular}
\end{table}

\begin{figure}[t]
\centering
\includegraphics[width=0.66\linewidth]{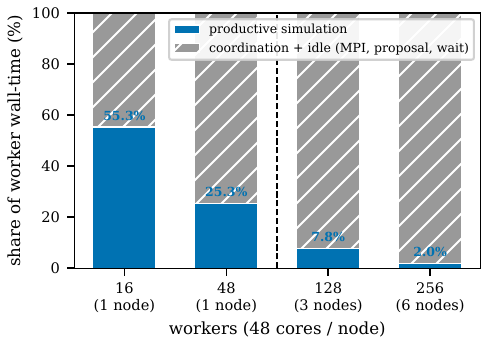}
\caption{Measured decomposition of asynchronous worker wall-time on Lotka--Volterra ($k=100$; $180$\,s budget) from a dedicated instrumented run. The productive simulation fraction (blue, annotated), the share of worker wall-clock spent inside the simulator, collapses from $55\%$ at $16$ workers to $2\%$ at $256$, while coordination and idle (MPI exchange, per-arrival proposal reconstruction, serialization, waiting) dominate as the job spreads beyond one node (16/48 = 1 node; 128 = 3 nodes; 256 = 6 nodes). The strong-scaling decline of Fig.~\ref{fig:scaling} is therefore a measured property of coordinating low-cost evaluations.}\label{fig:lv-timing}
\end{figure}

\subsection{Calibration against the archive size and the snapshot buffer}
The calibration tests of \S\ref{sec:results-posterior} fix $k$ and $S$ per benchmark, which leaves open whether the two settings must be retuned as the dimension grows. To test this, we sweep them directly on a $d$-dimensional Gaussian-mean target, $y_i\sim\mathcal N(\theta,\sigma^2 I_d)$ with $\theta\in\mathbb R^d$ under a flat box prior, whose posterior is known exactly at every $d$. Simulation-based calibration against that exact posterior returns $0.496/0.794/0.905/0.952$, so the target itself is calibrated and any departure below is a property of the method. The grid is $d\in\{1,2,4,8,16\}\times k\in\{50,100,200,400,800\}\times S\in\{0,5,20,50\}$, $1000$ trials per cell, reporting the same full-history estimator (Table~\ref{tab:ks-sweep}). $S$ sets both the online buffer length and the number of reconstructed proposals in the posterior denominator (\S\ref{sec:method-amis}), and the sweep varies the two together; at $S=0$ the denominator holds only the first and the last reconstructed proposal, so the $S=0$ column tests the posterior-side mixture as much as the online buffer. At $S=0$ the estimator is over-confident in one and two dimensions (mean deviation $-0.088$); every $S\ge5$ removes that deviation, and the experiments use $S=20$.

\paragraph{The archive size does not, in this range, trade calibration against throughput.} The archive size has both a statistical effect (above) and a systems cost (per-arrival proposal reconstruction is $\mathcal{O}(k)$, so a larger archive is slower at every worker count; \S\ref{sec:results-boundary}). Table~\ref{tab:k-frontier} puts the throughput cost of $k$ next to its worst-case miscalibration over the dimensions swept, using a throughput grid matched on everything but $k$ (Lotka--Volterra, $180$\,s wall cap, five replicates, worker counts $1$--$288$: $k\in\{100,1000\}$ from the main scaling run and $k\in\{50,200,400,800\}$ from a dedicated $k$-frontier run sharing that configuration). Of the six archive sizes, only the two smallest are Pareto-optimal. The setting $k=100$ is both the best-calibrated one tested (worst deviation $0.064$, at $d=4$) and within $7\%$ of the fastest ($6240$ vs $6684$ simulations/s at the single-node peak); every larger archive is dominated by it on both axes: $k=200$ is $5\%$ slower and nearly three times worse calibrated ($0.180$), $k=400$ is $14\%$ slower with a deviation of $0.195$, and $k=800$ costs $31\%$ of throughput and is still slightly worse calibrated than $k=100$ ($0.074$). The only faster archive, $k=50$, fails calibration ($0.250$ at $d=8$). Above that small-archive failure region, enlarging the archive adds no calibration and costs throughput monotonically.

The two axes come from different benchmarks by design (calibration needs a target whose posterior is known at every $d$; throughput needs the regime where an $\mathcal{O}(k)$ per-arrival cost is most exposed), so the throughput column is an upper bound on the relative cost of $k$. At scale, $k$ shifts the crossover with pyABC (population equal to the worker count) by a single step: pyABC overtakes beyond three fully packed nodes for $k\le200$ and at three nodes for $k\ge400$ (last column).

\begin{table}[t]
\caption{Calibration sensitivity to the archive size $k$ and the AMIS snapshot buffer $S$ on the $d$-dimensional Gaussian-mean target ($1000$ trials per cell; entries are the mean signed deviation of empirical from nominal coverage over the $0.50/0.80/0.90/0.95$ levels and the $d$ parameters). Negative is over-confident, positive is conservative; $0.000$ is exactly calibrated. \textbf{Left:} varying $k$ at $S=20$. \textbf{Right:} varying $S$ at $k=100$.}\label{tab:ks-sweep}
\centering
\begin{tabular}{@{}lccccc@{}}
\toprule
& \multicolumn{5}{c@{}}{archive size $k$ (at $S=20$)} \\
\cmidrule(l){2-6}
& 50 & 100 & 200 & 400 & 800 \\
\midrule
$d=1$  & -0.014 & -0.003 & +0.006 & +0.006 & +0.074 \\
$d=2$  & -0.020 & -0.006 & +0.009 & +0.030 & +0.072 \\
$d=4$  & +0.025 & +0.064 & +0.094 & +0.133 & +0.063 \\
$d=8$  & \textbf{-0.250} & +0.055 & +0.180 & +0.195 & +0.061 \\
$d=16$ & -0.047 & +0.028 & +0.049 & +0.053 & +0.055 \\
\bottomrule
\end{tabular}
\vspace{0.6em}

\begin{tabular}{@{}lcccc@{}}
\toprule
& \multicolumn{4}{c@{}}{snapshot buffer $S$ (at $k=100$)} \\
\cmidrule(l){2-5}
& 0 & 5 & 20 & 50 \\
\midrule
$d=1$  & \textbf{-0.088} & -0.001 & -0.003 & -0.005 \\
$d=2$  & \textbf{-0.088} & -0.001 & -0.006 & -0.003 \\
$d=4$  & +0.039 & +0.068 & +0.064 & +0.064 \\
$d=8$  & +0.065 & +0.006 & +0.055 & +0.074 \\
$d=16$ & -0.010 & +0.010 & +0.028 & +0.034 \\
\bottomrule
\end{tabular}
\end{table}

\begin{table}[t]
\caption{The archive size on both of its axes at once. \textbf{Throughput}: median asynchronous throughput at the single-node peak ($48$ workers) on Lotka--Volterra, $180$\,s wall cap, five replicates ($k\in\{100,1000\}$ from the main scaling run, $k\in\{50,200,400,800\}$ from the matched $k$-frontier run), and its change relative to the $k{=}100$ we use throughout. \textbf{Calibration}: worst absolute deviation of empirical from nominal coverage over the dimensions of the $d$-dimensional Gaussian-mean sweep at $S{=}20$ (Table~\ref{tab:ks-sweep}), with the dimension attaining it. \textbf{Crossover}: fully packed nodes at which pyABC (population $=$ worker count) overtakes the asynchronous sampler. The sweep did not include $k{=}1000$. Bold marks the best value per column.}\label{tab:k-frontier}
\centering
\begin{tabular}{@{}lccccc@{}}
\toprule
$k$ & sims/s at $W{=}48$ & vs.\ $k{=}100$ & worst $|\Delta\text{cov}|$ & at & pyABC overtakes \\
\midrule
50   & \textbf{6684} & {\boldmath$+7\%$}  & 0.250 & $d{=}8$ & \textbf{beyond 3 nodes} \\
100  & 6240 & ---     & \textbf{0.064} & $d{=}4$ & \textbf{beyond 3 nodes} \\
200  & 5904 & $-5\%$  & 0.180 & $d{=}8$ & \textbf{beyond 3 nodes} \\
400  & 5393 & $-14\%$ & 0.195 & $d{=}8$ & 3 nodes \\
800  & 4331 & $-31\%$ & 0.074 & $d{=}1$ & 3 nodes \\
1000 & 3804 & $-39\%$ & ---   & ---     & 3 nodes \\
\bottomrule
\end{tabular}
\end{table}

\begin{figure}[t]
\centering
\includegraphics[width=0.48\linewidth]{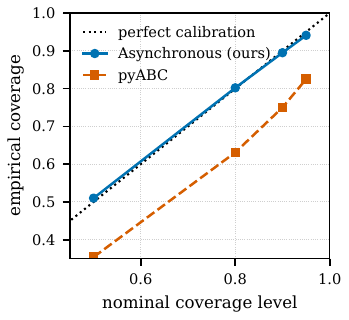}\hfill
\includegraphics[width=0.48\linewidth]{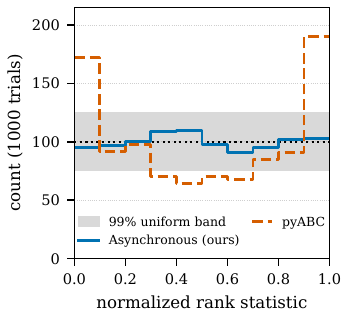}
\caption{Simulation-based calibration on Gaussian-mean inference (1000 trials; the full-history AMIS posterior): empirical coverage versus nominal level (left) and rank histogram (right). The asynchronous estimator is close to calibrated.}\label{fig:sbc}
\end{figure}

\subsection{Worker utilization on Cellular Potts at $50^3$}
\begin{figure}[t]
\centering
\includegraphics[width=0.52\linewidth]{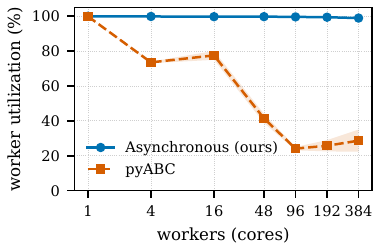}
\caption{Cellular Potts $50^3$ worker utilization, earlier parameterization (fraction of worker wall-clock spent inside the simulator) versus worker count ($1800$\,s budget; error bars are $\pm1$ standard deviation); pyABC uses a population equal to the worker count (Fig.~\ref{fig:scaling}, right). The asynchronous sampler stays ${\approx}98.9$--$99.8\%$ utilized at every scale. pyABC fills every core within a generation but then idles at the barrier waiting for the slowest simulation, so its utilization falls to ${\approx}24$--$29\%$ ($41\%\to24\%\to26\%\to29\%$ at $48/96/192/384$ workers, i.e.\ ${\approx}70$--$76\%$ idle at every multi-node scale). The figure is the measured utilization behind pyABC's sub-linear scaling in Fig.~\ref{fig:scaling} (right).}\label{fig:cpm-util}
\end{figure}

\subsection{Lotka--Volterra extinction}
The simulator has an absorbing state: if either population reaches zero before $T_{\max}$, the run returns a fallback discrepancy, which a smooth kernel weights to effectively zero. At the configured parameters that is the usual outcome (survival $2.6\%$ of $2000$ simulations at the true parameter, $1.7\%$ over $400$ uniform prior draws). The observed dataset is itself re-drawn until a trajectory survives, so the target is $p(\theta \mid s_{\mathrm{obs}}, \text{survived})$, which is internally consistent: the extinct particles' zero weight is the $P(\text{survive}\mid\theta)$ factor of the likelihood. Throughput and scaling are unaffected, because an extinct simulation costs what a surviving one costs. Posterior recovery is affected: the informative sample is ${\approx}2\%$ of the nominal budget, so we report Lotka--Volterra as a systems benchmark and treat its posterior curves as diagnostics that bound no method's attainable accuracy.

\end{document}